\documentclass[11pt]{article}

\usepackage[T1]{fontenc}
\usepackage{lmodern}
\usepackage[margin=1in]{geometry}
\usepackage{amsmath,amssymb,amsthm,mathtools}
\usepackage{bm}
\usepackage{microtype}
\usepackage{comment}
\usepackage[numbers]{natbib} 
\usepackage[hidelinks]{hyperref}
\usepackage[nameinlink,capitalize,noabbrev]{cleveref}
\usepackage{enumitem}
\usepackage[normalem]{ulem}
\usepackage{cancel}

\renewcommand{\xout}[1]{\ifmmode\cancel{#1}\else\sout{#1}\fi}

\usepackage{dsfont}
\usepackage{bbm} 
\usepackage[normalem]{ulem} 
\usepackage{xcolor} 
\usepackage{braket} 
\usepackage{float}  
\usepackage{lipsum} 

\allowdisplaybreaks 

\newcommand{\ketbra}[2]{\ket{#1}\!\!\bra{#2}}

\newcommand{\idd}[1]{1}%
\newcommand{\id}{{\mathbbm{1}}}

\newcommand\mpwS[1]{[MW:{\let\helpcmd\sout\parhelp#1\par\relax\relax}] }
\long\def\parhelp#1\par#2\relax{%
	\helpcmd{#1}\ifx\relax#2\else\par\parhelp#2\relax\fi%
}

\newcommand{\nocontentsline}[3]{}
\newcommand{\tocless}[2]{\bgroup\let\addcontentsline=\nocontentsline#1{#2}\egroup}

\newtheorem{theorem}{Theorem}
\newtheorem{proposition}[theorem]{Proposition}
\newtheorem{lemma}[theorem]{Lemma}

\theoremstyle{definition}
\newtheorem{definition}[theorem]{Definition}
\newtheorem{remark}[theorem]{Remark}

\newcommand{\C}{\mathbb C}
\newcommand{\R}{\mathbb R}
\newcommand{\Herm}{\operatorname{Herm}}

\newcommand{\cB}{\mathcal B}

\newcommand{\Gmap}{\Gamma}

\newcommand{\Aalg}{\mathcal A}

\newcommand{\Ct}{\mathcal{C}_{\mathrm{tensor}}^{\mathrm{fin}}}
\newcommand{\Cc}{\mathcal{C}_{\mathrm{commuting}}^{\mathrm{fin}}}
\newcommand{\CtR}{\mathcal{C}_{\mathrm{tensor},\mathbb R}^{\mathrm{fin}}}
\newcommand{\CcR}{\mathcal{C}_{\mathrm{commuting},\mathbb R}^{\mathrm{fin}}}
\newcommand{\CRtR}{\mathcal{C}_{\mathrm{R\text{-}tensor},\mathbb R}^{\mathrm{fin}}}
\newcommand{\TrX}[2]{\mathrm{Tr}_{#1}\left[ #2 \right]}

\title{Fermions and the Tsirelson Problem in Real Quantum Theory}
\author{%
Timoth{\'e}e Hoffreumon\textsuperscript{1}
\qquad
Mischa P. Woods\textsuperscript{2}\\[0.75em]
\small\textsuperscript{1}Mathematical Institute, Slovak Academy of Sciences,
Bratislava, Slovakia\\
\small\textsuperscript{2}ENS Lyon, Inria, Lyon, France}
\date{}

\begin{document}
\maketitle

\begin{abstract}
In a previous work, we showed that 
 real quantum theory (RQT) cannot be experimentally distinguished from ordinary 
quantum theory (QT). In this paper, we investigate whether this indistinguishability survives two natural extensions of the setting: allowing separated measurements to be represented by commuting rather than tensor-product observables, and imposing superselection rules, as in fermionic information theory. We show that it does in both cases.

Our first result is a positive answer to the Tsirelson problem in RQT: commuting and tensor-product measurements generate the same finite-dimensional correlations. This is surprisingly not a trivial consequence of the standard complex proof, as it does not carry over to the reals. The proof obstruction becomes particularly relevant in the presence of superselection rules, which is why we first establish the equivalence in RQT before extending the framework to superselected RQT.

Our second result is a local, real quantum formulation of superselected theories. We extend the RQT framework to accommodate a representation of superselection over the reals and show, in particular, that it admits a fermionic version whose predictions are experimentally indistinguishable from those of fermionic information theory. More generally, the construction extends to all superselected QTs.
\end{abstract}

\tableofcontents

\section{Introduction}

In quantum theory, a superselection rule is most naturally understood as a restriction on which preparations, measurements, and transformations are physically available, rather than as a replacement of the underlying system.  Fix a group $G$, a finite-dimensional complex Hilbert space $\mathcal H_A$ for each system $A$, and a unitary representation
\begin{equation}
	U_A:G\longrightarrow\mathrm U(\mathcal H_A).
\end{equation}
Writing
\begin{equation}
	\mathcal U_A(g)(X):=U_A(g)XU_A(g)^\dag\:,
\end{equation}
the associated superselected, or invariant, quantum theory permits precisely those states and effects that are fixed by this action and those transformations that are equivariant under it:
\begin{align}
	\mathcal U_A(g)(\rho)&=\rho,\qquad
	\mathcal U_A(g)(E)=E,\nonumber\\
	\Phi\circ\mathcal U_A(g)
	&=\mathcal U_B(g)\circ\Phi
	\qquad\text{for every }g\in G.
\end{align}
The last condition is imposed separately on every outcome subchannel of an instrument, and every effect of a permitted POVM is invariant. This is the operational ``invariant world'' formulation of a superselection rule used by Kitaev, Mayers, and Preskill \cite[Secs.~II.B--C]{KMP}.  For a state preparation, regarded as a channel from the one-dimensional trivial representation, the channel condition reduces exactly to the state condition above.

These definitions make clear that superselected theories are, in particular, sub-theories of quantum theory. If $\mathsf{QT}^{\mathrm{inv}}_G$ denotes the invariant theory and $\mathsf{QT}$ the unrestricted one, then
\begin{equation}
	\mathsf{QT}^{\mathrm{inv}}_G\subseteq\mathsf{QT}\:.
\end{equation}
Passing from the invariant theory to unrestricted quantum theory does not change any Hilbert space, matrix, channel, outcome, or classical control label: it simply forgets the rule that prohibited the non-invariant alternatives.  In this literal sense, it is ``obvious'' that a superselected quantum theory is embedded in ordinary quantum theory.  Every Born probability, every normalised or unnormalised conditional state, and every joint or conditional
probability in a finite adaptive protocol is therefore identical before and after the restriction is forgotten.  No reference frame or further simulation is needed; indeed, the reference-frame construction of Kitaev, Mayers, and Preskill sends an already invariant operator $M$
simply to $\id_R\otimes M$ \cite[Eq.~(30)]{KMP}. Throughout this document, $\otimes$ denotes the Kronecker product unless stated otherwise.

Recently, we~\cite{HoffreumonWoods2026} showed that quantum theory can be embedded within real quantum theory in a way that preserves the operational meaning of all operators, including their locality structure (i.e., over which subsystems they act non-trivially). If $\mathsf{QT}$ denotes (finite-dimensional) quantum theory and $\mathsf{RQT}$ the finite-dimensional theory in which states and effects are represented by real symmetric matrices, then
\begin{equation}
	\mathsf{QT} \subset \mathsf{RQT}.
\end{equation}
The construction is based on the standard embedding from the complexes to the reals. Write $A = \mathrm{Re}(A)+i\mathrm{Im}(A)$ and define
\begin{equation}
	\Gmap_d(A)	:=
	\begin{pmatrix}
		\mathrm{Re}(A)&-\mathrm{Im}(A)\\
		\mathrm{Im}(A)& \mathrm{Re}(A)
	\end{pmatrix}.
\end{equation}
This map sends Hermitian positive matrices to real symmetric positive matrices and doubles the trace. Consequently, the corresponding RQT states and effects are
\begin{equation}
	\widetilde\rho=\tfrac12\Gmap_d(\rho), \qquad	\widetilde E=\Gmap_d(E),
\end{equation}
so that their Born pairing is unchanged:
\begin{equation}
	\TrX{\R}{\widetilde\rho\,\widetilde E} =\TrX{\C}{\rho E} \:.
\end{equation}
For a physical symmetry, the same map sends each unitary $U_A(g)$ to an orthogonal matrix $O_A(g)=\Gmap(U_A(g))$. The quantum invariance and equivariance equations above are then carried into the analogous orthogonal-conjugation equations on the encoded real operators.

With the unitary group mapped to the orthogonal group, it seems that, at face value, the standard embedding not only provides a real representation of quantum theory within real quantum theory but also any superselected subtheory thereof. However, this is only true if locality is mathematically represented using the commuting-subsystems paradigm. If instead, the tensor-product-subsystems paradigm is used, a naive application of the standard mapping is insufficient as the tensor factors are not preserved through it,
\begin{equation}
	\Gmap_{d_Ad_B}(A \otimes B) \neq \Gmap_{d_A}(A) \otimes \Gmap_{d_B}(B) \:.
\end{equation}
The operational meaning of the composition is lost. Moreover, some interpretations~\cite{McKague,aleksandrova2013} of the embedding see it as entangling any complex qudit system with an auxiliary 2-dimensional RQT subsystem (a `rebit')
\begin{equation}
	\begin{aligned}
		\frac{1}{2}\Gmap_d(\rho)	= \frac{I}{2} \otimes \operatorname{Re}\rho + \frac{J}{2} \otimes \operatorname{Im}\rho \:,
	\end{aligned}
\end{equation}
where
\begin{equation}\label{eq:IandJ}
	I = \begin{pmatrix}
		1 & 0 \\ 0 & 1
	\end{pmatrix}\:, \qquad 
	J=\begin{pmatrix}0&-1\\1&0\end{pmatrix}\:.
\end{equation}
The auxiliary system is used to encode the complex phase in the reals; $\{I,J\}$ provides a representation of the complex unit group, with $I$ the identity representing multiplication by $1$ and $J$ the real quarter-turn representing multiplication by $i$.

Therefore, the nontrivial step is to pass from complex to real quantum theory while retaining the operational relations among the systems: one must find a way to extend the standard mapping such that these complex-phase-carrying subsystems are hidden in the representation and the complex Kronecker products $\otimes$ are properly mapped to their real counterparts. This is what we did in Ref.~\cite{HoffreumonWoods2026}. Yet, nothing guarantees that any superselection rule applied in the complex theory would be properly mapped in the real theory. It is natural to ask  whether bosonic or fermionic information theories can be mapped to the reals in a similar `locality-preserving' fashion, or whether their corresponding information theories would differ when based on the reals instead of the complex numbers. 

This is what we do in this note. 
Before developing the representation, we first needed to close a potential issue with our complex-to-real construction~\cite{HoffreumonWoods2026} by proving the validity of the analogue of Tsirelson conjecture over the reals. That is, we show that the commuting-subsystem paradigm is equivalent to the tensor-product-subsystems paradigm in finite dimensions. Indeed, while not relevant to the questions studied in our previous article, the validity of the conjecture could not reasonably be assumed when considering the tensor-product representation of group actions in the theory. 

In this respect, our approach is different from Ref.~\cite{YingEtAl2025Foil}. In that article, fermionic information theory and real quantum theory were presented as superselected quantum theories. In this article, we present a method for embedding superselected sub-theories of QT, such as FIT, within RQT, thereby lifting it to a superselected real quantum theory. The link between representations of complex-superselected theories in RQT and genuinely superselected subtheories of RQT is left for future research.

\section{Overview of results}
\subsection{A Tsirelson-style result for RQT}
One of the basic choices in formulating a bipartite QT is how to express the operational separation of Alice and Bob. In the usual tensor-product formulation, Alice and Bob's respective systems are assigned to distinct Hilbert-space factors, so that measurements have the form
\begin{equation}
	\Tilde{E}_a^x\otimes \id_B\:, \qquad \id_A \otimes \Tilde{F}_b^y\:,
\end{equation}
where $a$ and $x$ (respectively, $b$ and $y$) are Alice's (resp., Bob) outcome and setting.
For any joint state $\rho$, the observed behaviour, i.e. the joint distribution of the outcomes conditioned on the settings, is
\begin{equation}\label{eq:Tsi_tensor}
	p(a,b | x,y) = \TrX{}{\rho\left(\Tilde{E}_a^x\otimes \Tilde{F}_b^y\right)} \:.
\end{equation}
There is, however, a second natural way to encode separation: Alice and Bob's measurement operators may act on one common Hilbert space, but are required to commute,
\begin{equation}
	[{E}_a^x,{F}_b^y]=0 \qquad \text{for all }a,b,x,y \:.
\end{equation}
In that case, the probabilities are written as
\begin{equation}\label{eq:Tsi_commuting}
	p(a,b | x,y)=\TrX{}{\rho \, {E}_a^x{F}_b^y}.
\end{equation}
Tsirelson's problem asks whether these two meanings of separation give the same bipartite correlations (see e.g. \cite{ScholzWerner2008} for a review). 
In finite-dimensional (complex) QT, they do: every finite-dimensional commuting-operator strategy can be rewritten exactly as a finite-dimensional tensor-product strategy. This finite-dimensional result is standard (e.g. Theorem~1 in Ref.~\cite{ScholzWerner2008}). 

By requiring the states and effects to be represented by real symmetric matrices, one can ask the same question for RQT. Perhaps surprisingly, the same result does not carry directly over in RQT. That is, using the usual proof methods, there are pairs of commuting real POVMs $\{{E}_{a}^x\}_a$ and $\{{F}_b^y\}_b$ in \cref{eq:Tsi_commuting} that require to be represented by complex POVMs $\{\Tilde{E}_{a}^x\otimes \id\}_a$ and $\{\id \otimes \Tilde{F}_b^y\}_b$ in \cref{eq:Tsi_tensor}. Worse, this possibility is reinforced when the set of valid effects is linearly restricted, for example, when a superselection rule is imposed.
Therefore, it is not clear whether it is always possible to find a real, tensor-factor representation of commuting real POVMs, especially when they obey superselection rules. If the commuting definition produced a larger set of behaviours, it would amount to a different and potentially stronger variant of RQT, one in which ``local'' measurements were not exhausted by tensor-product local measurements.

The point of \cref{sec:Tsirelson} is to prove that the equivalence between commuting and tensor product local measurements still holds for RQT, independently of whether superselection rules are imposed or not. A minor caveat is that, compared to the complex case, the total matrix dimension of the real tensor-product representation is now larger than the commuting representation.\footnote{It is indeed minor because, in the usual setting where arbitrary local measurements are required, the optimization is performed over arbitrarily large finite dimensions, and this caveat evaporates.}

We start by giving a complete derivation of the complex case because the proof is short, useful for comparison with RQT, and a concise self-contained presentation is surprisingly hard to find in the literature.
Then, we show how to adapt the results of Ref.~\cite{HoffreumonWoods2026} in order to obtain a finite-dimensional equivalence between commuting and tensor product behaviours. It says that, for bipartite behaviours with finite setting and outcome sets, the tensor-product and commuting definitions of local measurements are equivalent in RQT, just as they are in finite-dimensional (complex) QT.  More precisely, \cref{prop:real-tsirelson} proves, in particular, that 
\begin{equation}
	\mathcal C_{\mathrm{tensor},\mathbb R}^{\mathrm{fin}} = \mathcal C_{\mathrm{commuting},\mathbb R}^{\mathrm{fin}} \:.
\end{equation}
Here, the left-hand side is the set of behaviours obtained from real symmetric states and tensor-product local real POVMs, while the right-hand side is the set obtained from real symmetric states and commuting real POVMs on one finite-dimensional real Hilbert space. 
By using the results of Ref.~\cite{HoffreumonWoods2026}, we also identify these real correlation sets with the ordinary finite-dimensional complex quantum ones\footnote{The detailed proof also records an extension to an $R$-tensor product representation introduced in Ref.~\cite{HoffreumonWoods2026}; this extra representation is not needed for the main overview statement.}, we collect all these results into \cref{thm:full_equivalence} which states that:
\begin{equation}
	\mathcal C_{\mathrm{tensor},\mathbb R}^{\mathrm{fin}}
	= \mathcal C_{\mathrm{commuting},\mathbb R}^{\mathrm{fin}} = \mathcal C_{\mathrm{tensor},\mathbb C}^{\mathrm{fin}} = \mathcal C_{\mathrm{commuting},\mathbb C}^{\mathrm{fin}}.
\end{equation}

The proof has two ingredients. First, the finite-dimensional $C^*$-algebra structure theorem known as Artin-Wedderburn (see e.g. Ref.~\cite{Mestoudjian_2025}) says that, after decomposing the Hilbert space into a finite direct sum of blocks, any algebra and its commutant act on opposite tensor factors:
\begin{equation}
	\mathcal H
	\cong
	\bigoplus_k
	\left(\mathcal H_{A,k}\otimes\mathcal H_{B,k}\right).
\end{equation}
Hence, commuting finite-dimensional measurements are already tensor-product measurements, up to a direct-sum label $k$ interpreted as classical correlations. 
Second, the locality-preserving realification of Ref.~\cite{HoffreumonWoods2026} reviewed above converts the blockwise complex tensor-product realisations into real symmetric ones without changing
any probabilities. The block label $k$ is then absorbed as shared classical randomness in a larger real tensor-product model. 

This statement is deliberately finite-dimensional. In (complex) QT, the corresponding infinite-dimensional assertion is false: the result $\mathrm{MIP}^*=\mathrm{RE}$ of Ji, Natarajan, Vidick, Wright, and Yuen implies a strict separation between tensor-product and commuting-operator correlation models \cite{JiNatarajanVidickWrightYuen2021}. It is therefore natural to expect that an infinite-dimensional commuting-measurement version of  RQT would also be more powerful than the corresponding tensor-product version, although we do not prove such an infinite-dimensional real separation here.

Finally, there is no equally direct ``dual'' theorem obtained by interchanging the roles of states and measurements. If both states and measurements were forced to be literal tensor products, then
\begin{equation}
	\TrX{}{(\rho_x^A\otimes\sigma_y^B) (E_a^x\otimes F_b^y)} = \TrX{}{\rho_x^A\,E_a^x}\: \TrX{}{\sigma_y^B\,F_b^y}\:,
\end{equation}
so all correlations would factorise. More importantly, RQT is not locally tomographic~\cite{wootters1990}: correlations present in a global real state need not be detectable by local real measurements~\cite{HoffreumonWoods2026}. Thus, a state-side analogue would require a different operational formulation, rather than a naive replacement of commuting measurements by commuting states.

\subsection{RQTs for superselected QTs}

The second question is whether imposing a superselection rule can distinguish QT from RQT. The inclusion of a superselected QT in unrestricted QT immediately gives a real representation of its probabilities, but does not by itself explain how the symmetry restrictions and the locality of operations appear in that representation. The point of \cref{sec:real-superselection-theories} is to make these features explicit. We first consider a single system of complex Hilbert-space dimension $d$: the realification map sends the unitary symmetry representation $U(g)$ to the orthogonal representation $O(g)=\Gamma(U(g))$, where $g\in G$ and $G$ is the superselection symmetry group introduced above, and carries invariant states and effects, together with equivariant channels, into corresponding real objects. More precisely, \cref{thm:algebra} identifies an embedded RQT subtheory with exactly the same Born probabilities. Writing $X$ for an encoded real state or effect, and $J_d:=J\otimes\id_d=\Gamma(i\id_d)$ for the real operator representing multiplication by $i$ (with $J$ defined in \cref{eq:IandJ} and $\id_d$ the $d$-dimensional identity), these states and effects satisfy both the original symmetry restriction and compatibility with the encoded complex structure,
\begin{equation}
    [X,O(g)]=0 \quad\text{for every }g\in G,
    \qquad [X,J_d]=0\:.
\end{equation}
The latter condition can itself be expressed as invariance under the cyclic group $C_4$ generated by $J_d$. Thus the embedded states and effects are selected by a combined $C_4\times G$ symmetry; invariance under $O(G)$ alone would generally allow a larger RQT theory.

We then turn to finite multipartite protocols, assuming that the symmetry acts on-site relative to the chosen tensor factors. Using the locality-preserving realification of Theorem~2 in Ref.~\cite{HoffreumonWoods2026}, \cref{thm:adaptive} shows that every such superselected QT protocol admits an RQT realisation with the same joint and conditional probabilities and the same prescribed locality of operations. The real representative of each operation depends only on that operation and its declared input and output blocks. Consequently, the construction applies consistently across protocols, rather than assigning a separate simulation to each experiment. As in the unrestricted case, preserving operational locality does not require product initial states to remain product states under the mapping.

Fermionic information theory provides the concrete example in \cref{sec:real-fermionic-quantum-theory}. Here $G=\mathbb Z_2$ implements parity superselection. We explicitly realify the Fock-space description, the parity operator, and the allowed states, effects, and channels, obtaining the same probabilities while retaining the graded commutation relations that specify fermionic locality. This yields an experimentally indistinguishable RQT representation of fermionic information theory (FIT). 

\section{A Tsirelson-style result for RQT}\label{sec:Tsirelson}

This section gives the self-contained derivation behind the overview statement. We first recall the two finite-dimensional models for ordinary (complex) QT, define the algebra generated by Alice's measurement operators and its commutant, and state the finite-dimensional $C^*$-algebra structure theorem.  We then prove the standard finite-dimensional equivalence between tensor-product and commuting-operator correlations.  Finally, we adapt the argument to RQT using the realification result of Ref.~\cite{HoffreumonWoods2026}. The finite-dimensional complex result is known \cite[Theorem~1]{ScholzWerner2008}; it is included here for completeness and because a short, explicit derivation clarifies exactly where the real-symmetric argument enters.

\subsection{The two notions of separation in QT}
Let $X,Y$ be finite sets of measurement choices (otherwise known as settings), and let $A,B$ be finite sets of outcomes. A bipartite behaviour is a table of conditional probabilities
\begin{equation}
    p(a,b| x,y), \qquad x\in X,\ y\in Y,\ a\in A,\ b\in B.
\end{equation}
A finite-dimensional quantum state is a positive operator $\rho$ of trace one. A positive operator-valued measure (POVM) for Alice, with measurement choice $x$, is an outcome-indexed family
\begin{equation}
    E_a^x\geq 0\:, \qquad \sum_{a\in A}E_a^x=\id.
\end{equation}
Likewise, Bob's POVM under setting $y$ is a family
\begin{equation}
    F_b^y\geq 0\:, \qquad \sum_{b\in B}F_b^y=\id.
\end{equation}

The finite-dimensional version of Tsirelson's problem compares two ways of expressing the separation of Alice and Bob; the exact finite-dimensional equivalence can be found in e.g., Theorem 1 of Ref.~\cite{ScholzWerner2008}.

\paragraph*{Tensor-product separation.} Alice and Bob have finite-dimensional Hilbert spaces $\mathcal H_A$ and $\mathcal H_B$.  Their joint space is $\mathcal H_A\otimes\mathcal H_B$, and their measurement operators act locally:
\begin{equation}
    \Tilde{E}_a^x\otimes \id_B\:, \qquad \id_A\otimes \Tilde{F}_b^y \:.
\end{equation}
For a state $\rho$ acting on $\mathcal H_A\otimes\mathcal H_B$,
\begin{equation}\label{eq:tensor_model}
    p(a,b| x,y)=\TrX{}{\rho\left(\Tilde{E}_a^x\otimes \Tilde{F}_b^y\right)}\:.
\end{equation}
Let 
\begin{equation}
    \Ct:= \{  \{ p(a,b| x,y) \}_{x\in X,\ y\in Y,\ a\in A,\ b\in B}  \}_{X,Y,A,B} 
\end{equation}
denote the set of all conditional probabilities which can be obtained in this way, allowing arbitrary finite dimensions for $\mathcal H_A$ and $\mathcal H_B$. Since we are in finite dimensions, all operators are matrices and $\otimes$ is the Kronecker product. 

\paragraph*{Commuting-operator separation.}
Instead, suppose that there is a single finite-dimensional Hilbert space $\mathcal H$.  Alice's and Bob's POVM elements act on that same space, but every effect of Alice commutes with every one of Bob:
\begin{equation}
    \forall a,b,x,y, \quad [E_a^x,F_b^y] := E_a^xF_b^y-F_b^yE_a^x = 0 \:.
\end{equation}
The set of conditional probabilities is
\begin{equation}\label{eq:commuting_model}
    p(a,b| x,y) = \TrX{}{\rho E_a^xF_b^y}\:,
\end{equation}
and let 
\begin{equation}
    \Cc := \{ \{   p(a,b\mid x,y) \}_{x\in X,\ y\in Y,\ a\in A,\ b\in B} \}_{X,Y,A,B} 
\end{equation}
denote the set of all conditional probabilities which can be obtained in this way, assuming an arbitrary but finite-dimensional underlying space $\mathcal{H}$.

Of course, every tensor-product strategy is immediately a commuting strategy, as
\begin{equation}
    [\Tilde{E}_a^x\otimes \id_B,\ \id_A\otimes \Tilde{F}_b^y]=0 \:.
\end{equation}
Thus
\begin{equation}
    \Ct\subseteq\Cc \:.
\end{equation}
But not much more can be immediately deduced.

\paragraph{Tsirelson problem.}
In finite-dimensional QT, however, the equivalence holds essentially all the way up to \cref{eq:Tsirelson_superstrong}. Before turning to the real quantum-theoretic case, we first prove \cref{eq:Tsirelson_strong} for (complex) QT. To this end, it suffices to prove the reverse inclusion, $\Ct \supseteq \Cc$. Presenting the standard quantum proof first will also make clear where the argument differs in RQT.

\begin{subequations}
A first question one may ask is whether these two situations are equivalent. More precisely, given finite sets $X,Y,A,B$, a finite-dimensional Hilbert space $\mathcal{H}$, and a state $\rho$, does one always have
\begin{equation}\label{eq:Tsirelson_superstrong}
    \TrX{}{\rho\left(\Tilde{E}_a^x\otimes \Tilde{F}_b^y\right)}= \TrX{}{\rho E_a^xF_b^y} \:?
\end{equation}
Of course, while this is a mathematically interesting question, it is too strong to have a meaningful physical interpretation. It requires trusting the dimension of the Hilbert space in both descriptions, as well as assuming that the preparations and measurements implemented by the experimentalists are faithfully represented by the corresponding operators.

A more physically meaningful approach is to assume that the only quantities the experimentalists can trust at the end of the experiment are $p(a,b|x,y)$, the observed distributions of the outcomes given the settings. This leads to the refined question of whether, for any collection of finite sets $X,Y,A,B$,
\begin{equation}\label{eq:Tsirelson_strong}
\Ct = \Cc \:?
\end{equation}
Here, the only fixed quantities are the sets of settings and outcomes. In particular, this question makes no reference to the dimensionality of the underlying Hilbert space. It may be that the commuting paradigm allows for more measurement strategies than the tensor-product paradigm, yet yields no distributions beyond those achievable with tensor-product measurements.

Finally, one can consider a further weakening of this notion of equivalence. Since it may not be possible to exhaustively characterise all Bell scenarios for a given source, one may instead focus only on the extremal values of Bell-like functionals. Consider a Bell-like functional on the set of distributions of the form, $B:\mathcal{C}\rightarrow\R^+$:
\begin{equation*}\label{eq:Tsirelson_normal}
B\{p(a,b|x,y)\} = \sum_{a,b,x,y} \mu_{a,b,x,y} p(a,b|x,y) \:,
\end{equation*}
and ask whether its supremum is the same when $\mathcal{C}$ is either $\Ct$ or $\Cc$. That is, given any collection of finite sets $X,Y,A,B$ and any collection of coefficients $\mu_{a,b,x,y}$, does one have
\begin{equation}
    \mathrm{sup}_{\substack{\mathcal{H}, \rho, \\\Tilde{E}_a^x, \Tilde{F}_B^y}} \left( \sum_{a,b,x,y}  \mu_{a,b,x,y} \TrX{}{\rho\left(\Tilde{E}_a^x\otimes \Tilde{F}_b^y\right)} \right) = \mathrm{sup}_{\substack{\mathcal{H}, \rho,\\ E_a^x, F_B^y}} \left(\sum_{a,b,x,y}  \mu_{a,b,x,y}\TrX{}{\rho\: E_a^x F_b^y} \right) \:?
\end{equation}
\end{subequations}
What we refer to as the Tsirelson problem, at least in this note, is only this last equality. This formulation gives the two notions of separated measurements even more room to diverge: the commuting paradigm may, in principle, produce statistics that cannot be reproduced by tensor-product measurements. The question is then whether this difference can be detected by Bell-like functionals, i.e., whether the two convex sets of correlations have the same boundary as seen through all such functionals.

In finite-dimensional QT, however, the equivalence holds almost up to \cref{eq:Tsirelson_superstrong} (the minor subtleties preventing strict unitary equivalence, arising from the sector decomposition, will be explained below)
Before turning to the real quantum-theoretic case, we first prove \cref{eq:Tsirelson_strong} for (complex) QT. To this end, it suffices to prove the reverse inclusion, $\Ct \supseteq \Cc$. Presenting the standard quantum proof first will also allow us to make clear where the proof differs in RQT.

\subsection{Reviewing the solution for finite-dimensional QT}
The proof of the Tsirelson conjecture in finite-dimensional QT relies on structural theorems from $C^*$-algebraic methods. In a nutshell, the set of all observables that Alice and Bob have access to is always embedded in the algebras of all complex operators, which is a $C^*$-algebra. Alice and Bob's operations are then embedded as two commuting sub-algebras, and the representation theory of such sub-algebras can be leveraged.

\subsubsection{The algebra generated by Alice and its commutant}
We use $(\cdot)^*$ to represent the adjoint of linear operators. In terms of matrices, this corresponds to the usual conjugate transpose.
\begin{definition}[$*$-algebra]
    A {\emph{(unital, finite-dimensional, complex) $*$-algebra $\mathcal{A}$}} is a finite-dimensional algebra over the complex numbers that contains the unit, and that is closed under an antilinear involution $^*$, that is a map $\mathcal{A} \rightarrow \mathcal{A}$ verifying the following: $\forall A,B \in \mathcal{A}$, $\forall c\in \C$:
    \begin{itemize}
        \item $(cA)^* = \overline{c} A$ ;
        \item $(A+B)^* = A^* + B^*$ ;
        \item $(AB)^* = B^* A^*$ .
    \end{itemize}
    Accordingly, a $*$-subalgebra is any $*$-closed unital subalgebra of a $*$-algebra.
\end{definition}

All algebras under consideration in this note will be assumed unital. 
In finite dimensions, a $C^*$-algebra is any algebra isomorphic to the $*$-algebra of complex operators on a Hilbert space or any $*$-subalgebra thereof.

\begin{definition}[$C^*$-algebra]
 A \emph{(unital, finite-dimensional, complex) $C^*$-algebra $\mathcal{A}$} is a $*$-algebra isomorphic to either the $*$-algebra $\mathcal{B}(\mathcal{H})$ of linear operators over a complex Hilbert space $\mathcal{H}$ under the usual operator product and the dagger as its involution, or any $*$-subalgebra of it.
\end{definition}
This definition glosses over all the usual analytic complications that come with the general definition of $C^*$-algebras, as we are in the finite-dimensional case; the usual issues with convergence or norm-closure of such algebras are simply not relevant when $d<\infty$.

\begin{definition}[The algebra generated by Alice]
Given all of Alice's POVM elements
\begin{equation}
    \mathcal E=\{E_a^x:x\in X,\ a\in A\},
\end{equation}
the notation
\begin{equation}
    \Aalg=C^*(\mathcal E)=C^*(\{E_a^x\}_{a,x})
\end{equation}
means the \emph{smallest unital $C^*$-subalgebra of $\mathcal B(\mathcal H)$ containing every $E_a^x$}.
Equivalently,
\begin{equation}
    C^*(\mathcal E)
    =
    \bigcap\left\{
       \mathcal D:
       \mathcal D\subseteq\mathcal B(\mathcal H)
       \text{ is a unital $C^*$-subalgebra and }
       \mathcal E\subseteq\mathcal D
    \right\}.
\end{equation}
\end{definition}
Alternatively, $\mathcal{A}$ can be seen as the algebra generated by all noncommutative $*$-polynomials of the operators $E_a^x$. That is, by any finite linear combinations of finite products made from $I$, the $E_a^x$, and their adjoints. Since they are POVM effects, each $E_a^x$ is positive, hence self-adjoint, so adjoining $(E_a^x)^*$ adds nothing. Furthermore, the identity operator is included by definition of a POVM. 

Thus $C^*(\mathcal E)$ is not generated by one ``seed'' POVM element: it is generated jointly by the entire indexed family $\{E_a^x\}_{a,x}$.

\begin{definition}[Commutant]
For a set $\mathcal S\subseteq\mathcal B(\mathcal H)$, its
\emph{commutant} is
\begin{equation}
    \mathcal S'
    =
    \{T\in\mathcal B(\mathcal H):TS=ST
      \text{ for every }S\in\mathcal S\}.
\end{equation}
In particular,
\begin{equation}
    \Aalg'
    =
    \{T\in\mathcal B(\mathcal H):TA=AT
      \text{ for every }A\in\Aalg\}.
\end{equation}
\end{definition}

The commutant will be used to characterise the set of all operators that commute with Alice's effects.  By linearity, it should be obvious that if $F\in\mathcal B(\mathcal H)$ commutes with every effect of Alice $E_a^x$, it also commutes with every polynomial thereof, and thus $F\in\Aalg'$. 

In particular, the commuting-operators assumption implies that $F_b^y\in\Aalg'$ for every $ b,y$. What might be less clear but follow by doing the same reasoning but starting from Bob's POVM effects, is that $\Aalg'$ is also a $C^*$-algebra. 
Alice's operators, therefore, determine the algebra in which Bob's operators are allowed to lie, namely $\Aalg'$, although they do not determine Bob's actual choice of POVMs.

The important bit for our purposes is that Alice's operations (her effects) lie in a $C^*$ algebra $\Aalg$, Bob's lie in the commutant of $\Aalg$, $\Aalg'$, and therefore their joint operations belong to the set $\Aalg \cup \Aalg'$, which is the smallest $C^*$-algebra in $\mathcal{B}(\mathcal{H})$ spanned by $\Aalg$ and $\Aalg'$.

\subsubsection{The finite-dimensional \texorpdfstring{$C^*$}-algebra structure theorems}

The relevant structure theorem has an abstract form and a concrete form. The standard reference for it is \cite[Theorem 11.2]{TakesakiI}, but see also the more recent and self-contained presentation of Ref.~\cite{Mestoudjian_2025}. A useful operator-algebra presentation of this decomposition and its commutant aimed specifically at quantum information theorists can be found in \cite[Section~IV]{HolbrookKribsLaflamme2004}.

The abstract form of the structure theorem is the reason behind our very pedestrian definition of a $C^*$ algebra in finite dimensions. 
\begin{theorem}[Finite-dimensional $C^*$-algebra structure theorem]
\label{thm:structure}
Let $\Aalg\subseteq\mathcal B(\mathcal H)$ be a unital
$C^*$-subalgebra and suppose that $\mathcal H$ is finite
dimensional.  Then, there are positive integers $n_1,\ldots,n_r$ such that
\begin{equation}
    \Aalg\cong\bigoplus_{k=1}^r M_{n_k}(\C) \:.
\end{equation}
\end{theorem}
Of course, with respect to the definition we provided, it is a cyclic statement. Keep in mind that we chose the definition exactly because this powerful theorem exists. In order not to derive this theorem out of the blue, there exist more abstract definitions of a $C^*$ algebra that do not rely on the identification with matrix algebras, which we invite the curious reader to consult.

Here $M_n(\C)$ denotes the $C^*$-algebra of all $n\times n$ complex matrices. The $\cong$ means an isomorphism of $C^*$-algebras. Meaning that there exists a mapping
\begin{equation}
    \Phi:\Aalg\longrightarrow \bigoplus_{k=1}^r M_{n_k}(\C) \:,
\end{equation}
that preserves the $C^*$-algebra structure (it is a $C^*$-homomorphism):
\begin{equation}
\begin{aligned}
    \Phi(A+B)&=\Phi(A)+\Phi(B),
    &
    \Phi(cA)&=c\Phi(A),\\
    \Phi(AB)&=\Phi(A)\Phi(B),
    &
    \Phi(A^*)&=\Phi(A)^*,\\
    \Phi(\id_{\mathcal H})&=
    \bigoplus_{k=1}^r \id_{n_k},
\end{aligned}
\end{equation}
and that is bijective. 

This abstract isomorphism is not, by itself, a change of basis on the given Hilbert space $\mathcal H$. The matrices on the right-hand side can have different sizes from the matrices by which $\Aalg$ originally acts; the space $\bigoplus_{k=1}^r M_{n_k}(\C)$ might not have the same dimension as the embedding space $\mathcal{B}(\mathcal{H}) \cong M_d{\C}$.  

The concrete form below uses the additional classification of finite-dimensional representations of the matrix algebras to explicitly build the isomorphism from a given $C^*$-algebra $\Aalg$ to $M_d(\C)$ by finding the dimensions $(m_1,...m_r)$ needed to `amplify' $\bigoplus_{k=1}^r M_{n_k}(\C)$ into $\bigoplus_{k=1}^r M_{n_k}(\C) \otimes M_{m_k}(\C)$ such that $d = \sum_{k=1}^r n_km_k$, and therefore that the resulting concrete identification \emph{is} implemented by a unitary change of basis when $\mathcal{H} = \C^d$. 

The integers $m_k$ are the representation multiplicities: the simple algebra $M_{n_k}(\C)$ occurs $m_k$ times in the given representation on $\mathcal H$. By doing so, the representation also finds a way to accommodate the representation of the commutant $\Aalg'$ of $\Aalg$ within the multiplicity space. 
In this form, it is usually referred to as (a special case of) the \textit{Artin-Wedderburn theorem}. 

\begin{theorem}[Artin-Wedderburn]\label{thm:AW}
For a $C^*$-algebra $\Aalg$ in $\mathcal{B}(\mathcal{H})$, there exists positive integers $m_1,\ldots,m_r$ and a unitary
\begin{equation}
    U:\mathcal H \longrightarrow \bigoplus_{k=1}^r \left(\C^{n_k}\otimes\C^{m_k}\right)
\end{equation}
such that
\begin{equation}
    U\Aalg U^* = \left\{ \bigoplus_{k=1}^r(X_k\otimes \id_{m_k}) :X_k\in M_{n_k}(\C)\right\}\:.
\end{equation}
Furthermore, its commutant $\Aalg'$ in $\mathcal{B}(\mathcal{H})$ verifies:
\begin{equation}
    U\Aalg'U^* =\left\{\bigoplus_{k=1}^r(\id_{n_k}\otimes Y_k) :Y_k\in M_{m_k}(\C)\right\}\:.
\end{equation}
\end{theorem}

Here is a quick summary of the proof (see one of the aforementioned sources for a full breakdown). The abstract statement is obtained by finding the $k$ minimal central projections of $\Aalg$, which is the largest set of orthogonal projectors $\{P_k\}_k$ that commute with all of $\Aalg$. Once these are found, each of them can be identified with the identity of a space $M_{n_k}(\C)$, $P_k \rightarrow \id_{n_k}$, where $n_k$ is the rank of $P_k$. Each resulting summand is a finite-dimensional simple $C^*$-algebra (also called a factor), which is why every such algebra is isomorphic to a full matrix algebra $M_{n_k}(\C)$: $\forall k, \forall A \in \Aalg: P_k A \neq 0$ then $A \rightarrow X_k \in M_{n_k}(\C)$. The representation as $M_{n_j}(\C)$ of the $j$-th factor in the decomposition is then usually referred to as the $j$-th \textit{sector}. With each element of $\Aalg$ belonging to at least one sector, the minimal projections resolve the identity as $\id_\mathcal{H} = \sum_k P_k$. This explains the form taken by the unit under $\Phi$, $\Phi(\id_{\mathcal{H}}) = \bigoplus_{k=1}^r \id_{n_k}$: the representation of the unit must preserve each individual sector.

A representation of $M_{n_k}(\C)$ on a finite-dimensional Hilbert space of dimension $m_kn_k$ is a direct sum of $m_k$ copies of its defining representation on $\C^{n_k}$. Thus, by grouping the $m_k$ equivalent copies, this produces the sector-specific action $X_k\otimes \id_{m_k}$ acting on $\C^{n_k}\otimes\C^{m_k}$. And by grouping the $k$ sectors, we recover the full space $\mathcal{H} = \bigoplus_{k=1}^r\C^{n_k}\otimes\C^{m_k}$. This yields the concrete (Artin-Wedderburn) form of the structure theorem.

To see why the commutant lives in the opposite tensor factor (in the multiplicity space), here is the matrix calculation underlying the last assertion of the structure theorem. In one sector,
\begin{equation}
    \left(M_n(\C)\otimes \id_m\right)' = \id_n\otimes M_m(\C) \:.
\end{equation}
To see why it is the case, write an arbitrary operator on $\C^n\otimes\C^m$ as $n^2$ $m\times m$ block matrices
\begin{equation}
    T=\sum_{p,q=1}^n e_{pq}\otimes T_{pq}\:, \qquad T_{pq}\in M_m(\C)\:,
\end{equation}
where $e_{pq}$ are the standard matrix units in $M_n(\C)$, and ask $T$ to commute with every element $X\otimes \id_{m}$ of $M_n(\C)\otimes \id_m$. Commutation with every diagonal matrix unit $e_{ii}\otimes \id_m$ forces $T_{pq}=0$ when $p\neq q$. Commutation with every $e_{ij}\otimes \id_m$ then forces all diagonal blocks to be equal:
\begin{equation}
    T_{11}=T_{22}=\cdots=T_{nn}=Y.
\end{equation}
Hence $T=\id_n\otimes Y$.

For the full direct sum, the central projection $P_k$ onto sector $k$ belongs to $\Aalg$.  Any $T\in\Aalg'$ commutes with every $P_k$, so it preserves each sector and has no off-diagonal blocks between distinct sectors. Applying the one-sector calculation inside each sector gives
\begin{equation}
    \Aalg' = \bigoplus_k \left(\id_{n_k}\otimes M_{m_k}(\C)\right) 
\end{equation}
in the basis supplied by $U$.

\subsubsection{Exact equivalence in finite dimensions}
The proof that the commuting and tensor sets of behaviours are equal then follow by applying these powerful structural theorems for $C^*$-algebras.
\begin{theorem}
\label{thm:equivalence}
For any collection of settings and outcomes sets $X,Y,A,B$, the sets of all possible distributions achievable under commuting and tensor-product complex measurements are equivalent, 
\begin{equation}
    \Ct=\Cc \:.
\end{equation}
\end{theorem}
\begin{proof}
We already proved $\Ct\subseteq\Cc$.  It remains to prove
$\Cc\subseteq\Ct$.

Take any finite-dimensional commuting-operator strategy on $\mathcal H$ (here identified with $\C^d$), under state $\rho$, Alice POVMs $\{E_a^x\}_{a,x}$, and Bob POVMs $\{F_b^y\}_{b,y}$. Define
\begin{equation}
    \Aalg=C^*(\{E_a^x\}_{a,x}).
\end{equation}
As explained above, every Bob operator belongs to $\Aalg'$.

Apply Theorem~\ref{thm:structure}. After the unitary change of basis
\begin{equation}
    \rho\mapsto U\rho U^*,
    \qquad
    E_a^x\mapsto UE_a^xU^*,
    \qquad
    F_b^y\mapsto UF_b^yU^*,
\end{equation}
and then suppressing the $U$'s from the notation, write
\begin{equation}
    \mathcal H
    =
    \bigoplus_{k=1}^r
       \left(\mathcal H_{A,k}\otimes\mathcal H_{B,k}\right),
    \qquad
    \mathcal H_{A,k}=\C^{n_k},
    \quad
    \mathcal H_{B,k}=\C^{m_k}.
\end{equation}
Because $E_a^x\in\Aalg$ and $F_b^y\in\Aalg'$, there are operators $\Tilde{E}_{a,k}^x\in\mathcal B(\mathcal H_{A,k})$ and $\Tilde{F}_{b,k}^y\in\mathcal B(\mathcal H_{B,k})$ such that
\begin{equation}
    E_a^x = \bigoplus_{k=1}^r \left( \Tilde{E}_{a,k}^x \otimes \id_{B,k}\right) \:,
    \qquad
    F_b^y = \bigoplus_{k=1}^r \left( \id_{A,k}\otimes \Tilde{F}_{b,k}^y \right) \:.
\end{equation}
Positivity of a block-diagonal operator is equivalent to positivity of every block. The POVM normalisation equations also hold block by block. Therefore, for every $k$,
\begin{equation}
    \Tilde{E}_{a,k}^x\geq0\: , \qquad \sum_a \Tilde{E}_{a,k}^x=\id_{A,k} \:, \qquad 
    \Tilde{F}_{b,k}^y\geq0, \qquad \sum_b \Tilde{F}_{b,k}^y=\id_{B,k}.
\end{equation}
Therefore, these define valid local POVMs for Alice and Bob in every sector $k$.

Let $P_k$ be the orthogonal projection onto the $k$-th sector, $\mathcal{H}_{A,k} \otimes \mathcal{H}_{B,k}$, and set
\begin{equation}
    q_k:=\TrX{}{\rho P_k} \:.
\end{equation}
Then $q_k\geq 0 $ and $\sum_k q_k=1$. For $q_k>0$, define the normalised state
\begin{equation}
    \rho_k=\frac{P_k\rho P_k}{q_k}
\end{equation}
on $\mathcal H_{A,k}\otimes\mathcal H_{B,k}$. (Sectors in which $q_k=0$ can be omitted.)

The product $E_a^xF_b^y$ is now represented by block-diagonal tensor products:
\begin{equation}
    E_a^xF_b^y = \bigoplus_k \left(\Tilde{E}_{a,k}^x\otimes \Tilde{F}_{b,k}^y\right) \:.
\end{equation}
Consequently, the off-diagonal blocks $P_k\rho P_\ell$ for $k\neq \ell$ make no contribution to the trace. Explicitly,
\begin{equation}
\begin{aligned}
    p(a,b\mid x,y) &= \TrX{}{\rho E_a^xF_b^y}\\
    &= \sum_{ijk\ell}\TrX{}{ (P_i\rho P_\ell) (P_k E_a^xF_b^y P_j)} \\
    & = \sum_{ik\ell} \TrX{}{ (P_i\rho P_\ell) (P_k \Tilde{E}_{a,k}^x\otimes \Tilde{F}_{b,k}^y P_k)}\\
    &= \sum_{k:q_k>0}q_k \TrX{}{\rho_k \left(\Tilde{E}_{a,k}^x\otimes \Tilde{F}_{b,k}^y\right)}\:.
\end{aligned}
\end{equation}
This already shows that the commuting strategy is a finite classical mixture of tensor-product strategies acting on a matching mixture of state preparations.

For completeness, we now show that these matching mixtures can be interpreted as the action of a single tensor-product strategy on a single state. Define the Alice- and Bob-specific Hilbert spaces,
\begin{equation}
    \widetilde{\mathcal H}_A =\bigoplus_k\mathcal H_{A,k} \:,
    \qquad
    \widetilde{\mathcal H}_B =\bigoplus_k\mathcal H_{B,k}\:,
\end{equation}
and the local POVMs
\begin{equation}
    \widetilde{E}_a^x=\bigoplus_k \Tilde{E}_{a,k}^x \:,
    \qquad
    \widetilde{F}_b^y=\bigoplus_k \Tilde{F}_{b,k}^y\: .
\end{equation}

We can describe the placement of the block states using projectors. Let $Q_{A,k}$ be the orthogonal projector from $\widetilde{\mathcal H}_A$ onto its summand $\mathcal H_{A,k}$, and let $Q_{B,k}$ be defined similarly. Then
\begin{equation}
    Q_k^{\mathrm{match}} = Q_{A,k}\otimes Q_{B,k}
\end{equation}
projects onto the matching sector $\mathcal H_{A,k}\otimes\mathcal H_{B,k}$ inside $\widetilde{\mathcal H}_A\otimes\widetilde{\mathcal H}_B$.  Under the natural identification of this sector with the original
$k$-th block, let $\overline{\rho}_k$ denote the operator that equals $\rho_k$ on that sector and that equals zero on its orthogonal complement. Equivalently, let
\begin{equation}
    Q_k^{\mathrm{match}}\overline{\rho}_k Q_k^{\mathrm{match}} =
    \overline{\rho}_k \:, \qquad\left.\overline{\rho}_k\right|_{\mathcal H_{A,k}\otimes\mathcal H_{B,k}} = \rho_k \:.
\end{equation}
The matching projectors are mutually orthogonal. We may therefore define
\begin{equation}
    \widetilde{\rho} = \sum_{k:q_k>0}q_k\overline{\rho}_k \:.
\end{equation}
This is a positive and trace-one matrix, hence it is a state. By construction, we then have that
\begin{equation}
    \TrX{}{\widetilde{\rho} \left(\widetilde{E}_a^x\otimes \widetilde{F}_b^y\right)} = \sum_{k:q_k>0}q_k \TrX{}{\rho_k\left(E_{a,k}^x\otimes F_{b,k}^y\right)} \:.
\end{equation}
Finally, by inserting into the previous equation, we get
\begin{equation}
    \begin{aligned}
        \TrX{}{\widetilde{\rho} \left(\widetilde{E}_a^x\otimes \widetilde{F}_b^y\right)} =p(a,b| x,y).
    \end{aligned}
\end{equation}

We see that the previously introduced projector $P_k$ could not be used directly in this last construction because $P_k$ acts on the original common space $\mathcal{H}$, keeping the interpretation of sectors as `a given preparation is measured by a corresponding measurement' apparent. On the other hand, $Q_k^{\mathrm{match}}$ acts on the new space $\widetilde{\mathcal H}_A\otimes\widetilde{\mathcal H}_B$ which hides away this sector-matching property: all preparations are randomised into a common mixture $\widetilde{\rho}$ and so are the measurements into the joint POVM $\{ \widetilde{E}_a^x\otimes \widetilde{F}_b^y\}$. In other words, the fact that each POVM effect has a block-decomposition such that each block acts non-trivially on exactly one of the components of the mixture $\widetilde{\rho}$ has been hidden. Nevertheless, the two projectors play the same role: each selects the sector $\mathcal H_{A,k}\otimes\mathcal H_{B,k}$ in its respective ambient Hilbert space.

Thus, the original commuting behaviour has an exact finite-dimensional tensor-product realisation. Therefore $\Cc\subseteq\Ct$, completing the proof.
\end{proof}

\subsubsection{Interpretation of the result}
The theorem identifies the two \emph{sets of finite-dimensional behaviors}; it does not assert that the original common Hilbert space $\mathcal H$ must itself be one tensor product on which Alice and Bob acts on opposite factors. In general, it decomposes as a direct sum of tensor products, i.e. a \textit{sectorised} tensor product,
\begin{equation}
    \mathcal{H} \cong \bigoplus_k \left(\mathcal H_{A,k}\otimes\mathcal H_{B,k}\right).
\end{equation}
The sector label $k$ behaves as shared classical randomness.  By placing the sector-$k$ states in the matching $k$-by-$k$ sectors of one larger tensor product, the same behaviour is reproduced exactly. No argument based on a limit, closure, or approximation has been involved.
The presence of these sectorial constraints is one of the reasons why Tsirelson's problem has been hard to solve in infinite dimensions: there, the existence of a set of minimal projections `indexing' each sector is a much more subtle problem \cite{ScholzWerner2008}. Often, this set is infinite and continuous, leading to many topological and functional-analysis issues when trying to reproduce the finite-dimensional proof. 

An unnecessary but interesting consequence of the theorem is that the dimensions of the commuting and tensor-product representations are the same. As we have seen in the proof, in the absence of sectorial constraints, the two pictures are unitarily equivalent, so \eqref{eq:Tsirelson_superstrong} is verified. When these are present, some bookkeeping is needed to interpret the sector-specific tensor-product strategies as commuting strategies. Hence, only \eqref{eq:Tsirelson_strong} is truly verified, albeit the required Hilbert-space dimensions are the same on both sides. 

This is where the RQT variant differs: while the sets of predicted correlations are the same, the Hilbert-space dimension needed to obtain one specific distribution may be greater in the commuting case.

\subsection{The finite-dimensional RQT variant}
We now impose the additional requirement that all states and effects in both models be represented by real symmetric matrices.  Throughout most of this section, `tensor product' means the tensor product inherited from the underlying (real) Hilbert space, which is the ordinary Kronecker product in the matrix representation. 
For our purposes, we will need another notion of `tensor product', the $R$/flag-product introduced in Refs.~\cite{Hoffreumon2025} and used as the key ingredient to derive the results of Ref.~\cite{HoffreumonWoods2026}. 

\subsubsection{Real models of separate measurements} 
Since RQT simply consists of substituting the complex Hilbert space with its real counterpart, all notions defined for QTs are directly transposed. Assuming a real Hilbert space $\mathcal{H}_\R$ of dimension $d$ finite, the main change is that a positive operator is represented by a real $d \times d$ matrix with non-negative spectrum.
A real quantum state remains a trace-one positive operator, while a real POVM remains a collection of positive operators summing up to the identity and indexed by an outcome sampled from a finite set. It is also assumed that the POVMs may actually be chosen from a collection of POVMs depending on a setting sampled from a finite set. 

For representational simplicity, we will assume that all dimensions are even in the following. In practice, every unknown dimension will be multiplied by a factor of 2, so $d$ becomes $2d$. This will simplify the comparison with the complex model, and it makes sense since, as inner product spaces, the dimension of a complex Hilbert space is twice that of the isomorphic real Hilbert space. 

Like in the complex case, there is a real tensor-product measurements model (we use different letters for the operators in order not to confuse them with their complex counterparts)
\begin{equation}
    p(a,b| x,y)=\TrX{}{\sigma\left(\Tilde{A}_a^x\otimes \Tilde{B}_b^y\right)}\:,
\end{equation}
where $\sigma $ is a state over the space $\mathcal{H}_A \otimes \mathcal{H}_B$ of dimension $2d_A2d_B$ and where $\{\Tilde{E}_a^x\}$ (respectively, $\{\Tilde{F}_b^y\}$) is a real POVM over the space $\mathcal{H}_A$ (resp., $\mathcal{H}_B$). So, for a fixed collection of finite sets $A,B,X,Y$, we define $\CtR$ 
to be the set of all conditional distributions which can be obtained from a pair of real quantum tensor-product measurements, for all possible dimensions, states, and pairs of tensor-product POVMS.

Similarly, there is a commuting measurements model 
\begin{equation}
    p(a,b| x,y) = \TrX{}{\sigma\: A_a^xB_b^y}, \qquad [A_a^x,B_b^y]=0,
\end{equation}
where $\{A_a^x\}_A, \{B_b^y\}_b \subset \mathcal{B}(\mathcal{H}_A \otimes \mathcal{H}_B)$ are two real quantum POVMs. We define the corresponding set of all conditional distributions as $\CcR$, 
considering again all possible dimensions, states, and pairs of commuting POVMs. 

Like with QT, it immediatly follows that
\begin{equation}\label{eq:CtR_in_CcR}
    \CtR \subseteq \CcR \:.
\end{equation}
Notice how RQT can always be seen as a restriction of QT in which one only uses the operators that are expressed as real matrices in a given basis. Because of that, it follows that
\begin{equation}\label{eq:CtR_in_Ct}
    \CtR \subseteq \Ct\:,
\end{equation} 
and 
\begin{equation}\label{eq:CcR_in_Cc}
     \CcR \subseteq \Cc\:.
\end{equation}

\paragraph{The R-product measurements.} Finally, we need to introduce a third model that will help us in our purposes. This model happens to be inferred from the image of the complex tensor-product model within the real tensor-product model. We only introduce the properties we need from it and refer the interested reader to Ref.~\cite{HoffreumonWoods2026}.
Here, the local spaces also have dimensions $2d_A$ and $2d_B$ but are composed under a different kind of `tensor product', which is the Segre embedding obtained by viewing the underlying space not as a real Hilbert space but as the representation of a complex projective space over a real Hilbert space. This bipartite composition rule is the $R$-product~\cite[Definition~A.9]{HoffreumonWoods2026} defined by
\begin{equation}
    A\otimes_R B = P_{AB}(A\otimes B)P_{AB} \:.
\end{equation}
with $P_{AB}$ an orthogonal projector which decomposes as
\begin{equation}
    P_{AB} := \frac{1}{2}\left(I\otimes \id_{d_A} \otimes I\otimes \id_{d_B} - J \otimes \id_{d_A} \otimes J \otimes \id_{d_B} \right) \:,
\end{equation}
where $I = \id_2$ and $J = \left(\begin{smallmatrix}
    0 & -1 \\ 1 & 0
\end{smallmatrix} \right)$. 
Besides this modification of the tensor product, the distributions are obtained in the same way as the tensor-product measurements model:
\begin{equation}
    p(a,b| x,y) = \TrX{}{\sigma\left(\Tilde{A}_a^x\otimes_R \Tilde{B}_b^y\right)}\:.
\end{equation}
We define the corresponding set of all conditional distributions as $\CRtR$, once again considering all possible dimensions, states, and pairs of commuting POVMs. By definition of the $R$-product, it should be clear that 
\begin{equation}\label{eq:CRtR_in_CtR}
    \CRtR \subseteq \CtR
\end{equation}

There are several remarkable properties of the $R$-product measurements model that will be needed in the following. First of all, since
\begin{equation}
    \sum_{a,b} A_a^x\otimes_R B_b^y = P_{AB}\:,
\end{equation}
the normalization of the probabilities is equivalent to
\begin{equation}
    \TrX{}{(\sigma P_{AB})}=1 \:.
\end{equation}
For a positive trace-one $\sigma$, this is the same as saying that $\sigma$ is supported on the range of $P_{AB}$.

The second property is that the distributions obtained from POVMs which effects respectively obey
\begin{subequations}
\begin{equation}
    [A_a^x,J \otimes \id_{d_A}] = 0\:, \qquad [B_b^y,J \otimes \id_{d_B}] = 0\:, 
\end{equation}
and whose state obey
\begin{equation}
    [\sigma, J_{AB}] = 0 \:,
\end{equation}
with $ J_{AB} := \frac{1}{2}\left( I \otimes \id_{d_A} \otimes J\otimes \id_{d_B} + J \otimes \id_{d_A} \otimes I \otimes \id_{d_B} \right)$, are in one-to-one correspondence with the distributions obtained from complex tensor-product commuting measurements over a complex Hilbert space of dimension $d_Ad_B$~\cite[Proposition 5]{HoffreumonWoods2026}.
\end{subequations}

The identification of the complex matrices leading to $\Ct$ within a subset of the reals matrices leading to $\CRtR$ is constructed via the triple of maps 
\begin{subequations}\label{eq:realification}
\begin{equation}
    (M,R_A,R_B): M_{d_Ad_B}(\C) \times M_{d_A}(\C) \times  M_{d_B}(\C) \longrightarrow M_{2d_A2d_B}(\R) \times M_{2d_A}(\R) \times  M_{2d_B}(\R) \:,
\end{equation}
defined by

    \begin{align}
        & R(X) = I \otimes \mathrm{Re}(X) + J \otimes \mathrm{Im}(X) \:;\\
        & M_{AB}(\tau) \cong \frac{1}{4}\left[ ( I\otimes I - J \otimes J) \otimes \mathrm{Re}(\tau) + (I \otimes J + J \otimes I) \otimes \mathrm{Im}(\tau) \right] \:.
    \end{align}
\end{subequations}
(The second equality is true only up to an inversion of the second and third tensor factors when $M_{AB}(\tau)$ is seen as a matrix acting on $\R^2 \otimes \R^2 \otimes \R^{d_A} \otimes \R^{d_B}$, hence the $\cong$ sign.)
These three mappings respectively identify valid complex states, local effects of Alice, and local effects of Bob, with valid real states, local effects of Alice, and local effects of Bob in such a way that all distributions are preserved. 
That is, for any finite-dimensional complex state $\tau$ and any two finite families of local complex POVMs $\{G_a^x\}$ and $\{H_b^y\}$, $M(\tau)$ is a real state, the two image families $\{R_A(G_a^x)\}_x$ and $\{R_B(H_b^y)\}$ are real POVMs, and
\begin{equation}
    \TrX{}{\tau\left(G_a^x\otimes H_b^y\right)} =
    \TrX{}{M(\tau) \left(R_A(G_a^x)\otimes_R R_B(H_b^y)\right)} \:.
\end{equation}

This implies that 
\begin{equation}\label{eq:Ct_in_CRtR}
    \Ct \subseteq \CRtR \:.
\end{equation}
Because $\CRtR \subseteq \CtR \subseteq \Ct$, we obtain a first equality:
\begin{equation}\label{eq:Ct_is_CRtR}
    \Ct = \CRtR \:.
\end{equation}
This might look strange: on the one hand, any distribution in $\CRtR$ can be seen trivially as one in $\Ct$. On the other hand, to see a distribution in $\Ct$ as one in $\CRtR$, the above construction should be used, and so the dimension must have been doubled. However, this is still an equivalence because the dimension is an unbounded but finite parameter: the $R$-product distributions for $2d_A2d_B$ dimensions are a subset of the tensor product distributions, which are a subset of the $R$-product distributions for $4d_A4d_B$ dimensions, which are a subset of the tensor product distributions in $4d_A4d_B$... and so on \textit{ad infinitum}.

Finally, the last property is that $R$-product measurements on states supported on the range of $P_{AB}$ produce exactly the same distributions as the corresponding tensor-product measurements~\cite[Proposition~6]{HoffreumonWoods2026}, i.e.,
\begin{equation}
    \TrX{}{M(\tau)\left(R_A(G_a^x)\otimes_R R_B(H_b^y)\right)} = \TrX{}{M(\tau)\left(R_A(G_a^x)\otimes R_B(H_b^y)\right)} \:,
\end{equation}
which in particular implies that~\cite[Eq.~(15)]{HoffreumonWoods2026}
\begin{equation}
    \TrX{}{\tau\left(G_a^x\otimes H_b^y\right)} = \TrX{}{M(\tau)\left(R_A(G_a^x)\otimes_R R_B(H_b^y)\right)}
\end{equation}
for every $a,b,x,y$. This recovers the well-known result that every distribution obtained with complex tensor-product measurements can be obtained using real tensor-product measurements instead~\cite{Pal_2008,Moroder_2013},
\begin{equation}\label{eq:Ct_in_CtR}
    \Ct \subseteq \CtR \:.
\end{equation}
Comparing it \cref{eq:CtR_in_Ct}, we obtain that the set of all finite, tensor-product distributions do not depend on the field, 
\begin{equation}\label{eq:Ct_is_CtR}
    \Ct = \CtR \:.
\end{equation}
This last property illustrates the usefulness of the $R$-product model: it provides a bridge between the complex and real tensor-product models.

\subsubsection{Proving the equivalence between the real commuting and real tensor product measurements}
We now prove that, also in finite-dimensional RQT, the commuting measurements and tensor-product measurements produce the same distributions. 
\begin{proposition}[Real finite-dimensional Tsirelson equivalence]
\label{prop:real-tsirelson}
For any collection of settings and outcomes sets $X,Y,A,B$, the sets of all possible distributions achievable under commuting, tensor-product, or $R$-product real measurements are equivalent, 
\begin{equation}\label{eq:CtR_is_CcR}
    \CtR = \CcR \:.
\end{equation}
The equality concerns sets $X,Y$ and $A,B$ and permits an increase in the local dimensions of a realising strategy.
\end{proposition}

\begin{proof}
As in the complex case, one direction of the inclusion, \cref{eq:CtR_in_CcR}, is immediate. It remains to prove that $\CtR \supseteq \CcR$. 

Take a conditional distribution in $\mathcal C_{\mathrm{commuting},\mathbb R}^{\mathrm{fin}}$, realized by a real symmetric state $\sigma \in M_{d_Ad_B}(\R)$ and real symmetric commuting POVM families $\{A_a^x\} \subset M_{d_A}(\R)$ and $\{B_b^y\}\subset M_{d_B}(\R)$. Regard these real matrices as complex matrices and apply exactly the structure-theorem and block-decomposition steps in the proof of Theorem~\ref{thm:equivalence}.  

Following that proof, after applying the unitary to the state and every effect, we reuse the same notation for the transformed matrices. This gives,
\begin{equation}
    \mathcal H = \bigoplus_{k=1}^r \left(\mathcal H_{A,k}\otimes\mathcal H_{B,k}\right)\:,
    \qquad
    \dim_{\mathbb C}\mathcal H_{A,k}=n_k\:,
    \quad
    \dim_{\mathbb C}\mathcal H_{B,k}=m_k\:,
\end{equation}
and
\begin{equation}
    A_a^x = \bigoplus_k \left(\Tilde{A}_{a,k}^x\otimes \id_{B,k}\right)\:,
    \qquad
    B_b^y = \bigoplus_k \left(\id_{A,k}\otimes \Tilde{B}_{b,k}^y\right).
\end{equation}

However, the key limitation of applying the Artin-Wedderburn theorem is that the unitaries it supplies need not be real. (The reasons behind this limitation are discussed below.)

Therefore the block matrices $\Tilde{A}_{a,k}^x$, $\Tilde{B}_{b,k}^y$, and the block states introduced below need not themselves have real entries. This is precisely where the realification result from Ref.~\cite{HoffreumonWoods2026} is needed.

Let $P_k$ be the projection onto the $k$-th block and, exactly as in the preceding proof, define
\begin{equation}
    q_k=\TrX{}{\sigma P_k}\:, \qquad \rho_k=\frac{P_k\sigma P_k}{q_k} \quad\text{when }q_k>0 \:.
\end{equation}
Starting from $ p(a,b| x,y) = \TrX{}{\rho A_a^x B_b^y}$, the calculation in the proof of Theorem~\ref{thm:equivalence} gives
\begin{equation}
    p(a,b| x,y) = \sum_{k:q_k>0}q_k \TrX{}{\rho_k\left(\Tilde{E}_{a,k}^x\otimes \Tilde{F}_{b,k}^y\right)}\:,
\end{equation}
where $\rho_k$, $\Tilde{E}_{a,k}^x$ and $\Tilde{F}_{b,k}^y$ may not be real matrices for all $k,a,x,b,y$'s. 

For each $k$ with $q_k>0$, apply the bipartite realification procedure described in \cref{eq:realification} to the state $\rho_k$ and to the complete POVM families $\{\Tilde{E}_{a,k}^x\}_{a,x}$ and $\{\Tilde{F}_{b,k}^y\}_{b,y}$. Denote the corresponding dimension-dependent maps by $M_k,R_{A,k},R_{B,k}$, and set
\begin{equation}
    \sigma_k=M_k(\rho_k), \qquad \Tilde{A}_{a,k}^x=R_{A,k}(\Tilde{E}_{a,k}^x),
    \qquad
    \Tilde{B}_{b,k}^y=R_{B,k}(\Tilde{F}_{b,k}^y).
\end{equation}
Then $\sigma_k$ is a real state ($\sigma_k \geq 0$ and $\TrX{}{\sigma_k} = 1$) on $\mathbb R^{2n_k}\otimes\mathbb R^{2m_k}$, whereas 
\begin{equation}
    \Tilde{A}_{a,k}^x\geq0\: , \qquad \sum_a A_{a,k}^x=\id_{2n_k} \:,
\end{equation}
and
\begin{equation}
    \Tilde{B}_{b,k}^y\geq0\:, \qquad \sum_b \Tilde{B}_{b,k}^y=\id_{2m_k} \: .
\end{equation}
For every sector $k$, settings $x,y$ and outcomes $a,b$, we know that~\cite[Proposition~6]{HoffreumonWoods2026}
\begin{equation}
    \TrX{}{\rho_k \left(\Tilde{E}_{a,k}^x\otimes \Tilde{F}_{b,k}^y\right)}
    =
    \TrX{}{\sigma_k \left(\Tilde{A}_{a,k}^x\otimes \Tilde{B}_{b,k}^y\right)}\:.
\end{equation}

We now combine the realified blocks, mirroring the final step of the proof of \cref{thm:equivalence}. Define the real local Hilbert spaces
\begin{equation}\label{eq:Hhat}
    \widehat{\mathcal H}_A
    =
    \bigoplus_k\mathbb R^{2n_k}\:,
    \qquad
    \widehat{\mathcal H}_B
    =
    \bigoplus_k\mathbb R^{2m_k}\:,
\end{equation}
and the real local POVMs
\begin{equation}
    \widehat{A}_a^x = \bigoplus_k \Tilde{A}_{a,k}^x,
    \qquad
    \widehat{B}_b^y = \bigoplus_k \Tilde{B}_{b,k}^y.
\end{equation}
(If a zero-weight block was omitted above, choose any real symmetric POVM on that block; it has no effect on the behaviour anyway.) The blockwise positivity and normalisation equations show that $\{\widehat{A}_a^x\}_a$ and $\{\widehat{B}_b^y\}_b$ are valid real symmetric POVMs. 

As in the preceding proof, we use projectors to place each block state in the new tensor-product space.  Let $\widehat Q_{A,k}$ be the orthogonal projector from $\widehat{\mathcal H}_A$ onto its summand $\mathbb R^{2n_k}$, and let $\widehat Q_{B,k}$ project $\widehat{\mathcal H}_B$ onto $\mathbb R^{2m_k}$. Then
\begin{equation}
    \widehat Q_k^{\mathrm{match}}
    =
    \widehat Q_{A,k}\otimes\widehat Q_{B,k}
\end{equation}
projects onto the matching realified sector $\mathbb R^{2n_k}\otimes\mathbb R^{2m_k}$ inside $\widehat{\mathcal H}_A\otimes\widehat{\mathcal H}_B$. Let $\overline{\sigma}_k$ be the operator that equals $\sigma_k$ on this sector and is zero on its orthogonal complement.  That is, 
\begin{equation}
    \widehat Q_k^{\mathrm{match}}\overline{\sigma}_k
       \widehat Q_k^{\mathrm{match}}
    =
    \overline{\sigma}_k,
    \qquad
    \left.\overline{\sigma}_k\right|_
       {\mathbb R^{2n_k}\otimes\mathbb R^{2m_k}}
    =
    \sigma_k.
\end{equation}
The projectors $\widehat Q_k^{\mathrm{match}}$ are real, symmetric, and mutually orthogonal. We therefore define
\begin{equation}
    \widehat{\sigma} = \sum_{k:q_k>0}q_k\overline{\sigma}_k.
\end{equation}
Every $\overline{\sigma}_k$ is real symmetric and positive, and the different terms have mutually orthogonal supports.
It follows that
\begin{equation}
    \widehat\sigma\geq0\:,
    \qquad
    \TrX{}{\widehat\sigma}=\sum_kq_k=1\:.
\end{equation}
Hence $\widehat\sigma$ is a valid real state.

These matching projectors are different from those in the preceding proof because realification has doubled each local block dimension: $\widehat Q_k^{\mathrm{match}}$ acts on $\mathbb R^{2n_k}\otimes\mathbb R^{2m_k}$, whereas $Q_k^{\mathrm{match}}$ acts on $\mathcal H_{A,k}\otimes\mathcal H_{B,k}$.  They nevertheless play the same role of selecting the $k$-th matching sector.

Finally, the direct-sum support of $\widehat\sigma$ makes $\TrX{}{ \widehat\sigma \left(\widehat{A}_a^x\otimes\widehat {B}_b^y\right)}$ equal to $\sum_k  q_k\TrX{}{\sigma_k \left(\Tilde{A}_{a,k}^x\otimes \Tilde{B}_{b,k}^y\right)}$. Grouping all key steps of the proof into one equation yields:
\begin{equation}
    \begin{aligned}
        \TrX{}{\sigma\: A_a^xB_b^y} =& 
    \sum_{k:q_k>0}q_k
    \TrX{}{\rho_k \left(\Tilde{E}_{a,k}^x\otimes \Tilde{F}_{b,k}^y\right)}  \\
    =&
    \sum_{k:q_k>0}q_k
    \TrX{}{\sigma_k \left(\Tilde{A}_{a,k}^x\otimes \Tilde{B}_{b,k}^y\right)}\\
    =&\TrX{}{ \widehat\sigma \left(\widehat{A}_a^x\otimes\widehat{B}_b^y\right)}.
    \end{aligned}
\end{equation}
We have therefore established that $\CcR \subseteq\CtR$, combining it with the reverse inclusion leads to the sought relation, \cref{eq:CtR_is_CcR}, which concludes the proof.
\end{proof}

With this proposition, we see that the three real models of separate measurements lead to the same set of correlations,
\begin{equation}
    \boxed{ \CRtR = \CtR = \CcR \:.}
\end{equation}
It follows, therefore, that Tsirelson's problem has a positive answer in finite-dimensional RQT.  Before concluding this section by comparing the real quantum correlations to their complex counterparts, some comments about the proofs are in order. We phrase each of these as a separate remark for ease of reference.

\begin{remark}[Difference with the complex case]
    The main difference with the complex case is the dimension mismatch. When proving that $\Cc \subseteq\Ct$, we constructed tensor product-measurements over the same Hilbert space as the commuting measurements. I.e., we showed that, given a state $\rho \in M_{d_Ad_B}(\C)$, and two collections of pairwise commuting POVMs $\{E^x_a\}_a^x$ and $\{F^y_b\}_b^y$, both in $M_{d_Ad_B}(\C)$ as well, there always exist a state $\widetilde{\rho} \in M_{d_Ad_B}(\C)$ and two collections of commuting POVMs $\{\widetilde{E}^x_a\}_a^x \subset M_{d_A}(\C)$ and $\{\widetilde{F}^y_b\}_b^y \subset M_{d_B}(\C)$, such that 
    \begin{equation}
        \TrX{}{\rho\: E^x_aF^y_b} = \TrX{}{\widetilde{\rho}\left( \widetilde{E}^x_a \otimes \widetilde{F}^y_b \right) } \:.
    \end{equation}
    This identification is almost an isomorphism between the two models, with a minor obstruction arising from the potential presence of sectors that must be interpreted as an additional source of randomness in $\Ct$. 

    Differently, when proving that $\CcR \subseteq \CtR$, we constructed tensor product-measurements over a Hilbert space four times the dimension of the one needed to represent the commuting measurements. I.e., we showed that, given a state $\sigma \in M_{d_Ad_B}(\R)$, and two collections of pairwise commuting POVMs $\{A^x_a\}_a^x$ and $\{B^y_b\}_b^y$, both in $M_{d_Ad_B}(\R)$ as well, there always exist a state $\widehat{\sigma} \in M_{2d_A2d_B}(\R)$ and two collections of commuting POVMs $\{\widehat{A}^x_a\}_a^x \subset M_{2d_A}(\R)$ and $\{\widehat{B}^y_b\}_b^y \subset M_{2d_B}(\R)$, such that 
    \begin{equation}
        \TrX{}{\sigma\: A^x_aB^y_b} = \TrX{}{\widehat{\sigma}\left( \widehat{A}^x_a \otimes \widehat{B}^y_b \right) } \:.
    \end{equation}
    This identification is clearly an embedding, in contrast to the almost-isomorphism in the complex case. 
    Therefore, the equivalence between $\CtR$ and $\CcR$ is more akin to the one between $\Ct$ and $\CRtR$: matching the predicted distributions imposes an ever-increasing dimension and thus an ever-increasing number of predicted distributions.
\end{remark}

\begin{remark}[Why the proof for the complex case could not be directly applied in the real case]
The obstruction to applying the Artin-Wedderburn decomposition in exactly the same way as in (complex) QT comes from the structure of finite-dimensional real $C^*$-algebras. Whereas every finite-dimensional complex $C^*$-algebra decomposes as a direct sum of simple complex factors, each isomorphic to some $M_n(\C)$, a finite-dimensional real $C^*$-algebra may contain simple factors of three different types, $M_n(\R), M_n(\C), M_n(\mathbb H)$, such that the abstract decomposition theorem for real $C^*$-algebras reads 
\begin{equation}
    \mathcal{A}_\R = \left(\bigoplus_k M_{n_k}(\R)\right) \oplus \left(\bigoplus_j M_{n_j}(\C)\right) \oplus \left(\bigoplus_\ell M_{n_\ell}(\mathbb{H})\right) \:.
\end{equation}

This distinction ultimately stems from the fact that $\R$ is not algebraically closed; a real representation may carry an additional complex or quaternionic structure that can only be irreducibly represented over the corresponding number field.

For a commuting pair of real factor $C^*$-subalgebras (i.e. both are isomorphic to $M_d(\R)$), the usual tensor-factor construction can still be carried out after a suitable choice of representation; the total dimension is guaranteed to be at least $d^2$, so the underlying space can be decomposed into $\R^d \otimes \R^d$ and each factor will only act on one side of the tensor product. The situation is more subtle when the algebras are sectorised or contain factors not isomorphic to $M_d(\R)$ but rather to $M_d(\C)$ or $M_d(\mathbb{H})$. In particular, the direct-sum decomposition into sectors need not admit a single tensor-product interpretation on the original real Hilbert space, as the tensor product of the abstract factors may be impossible to faithfully represent on the original space.

This phenomenon is particularly relevant for real superselected subtheories. Superselection constraints can restrict the observable algebra to a real subalgebra whose centre contains, for example, a nontrivial complex structure, namely an element $J$ satisfying $J^2=-\id$. The factor corresponding to this central element (the one selected by the projector associated with the central element $J$) is then automatically a complex matrix algebra rather than a real one. In this case, it might not be possible to represent it faithfully without first increasing the dimension of the overall space. Quaternionic structures can arise analogously. 

Thus, although all effects are represented by real matrices, the real $C^*$-algebra they span may not be a direct sum of purely real factors. To illustrate this, a simple example is provided by
\begin{equation}
    \mathcal A=\operatorname{span}_{\R}\{I,J\} \subset M_2(\R), \qquad \text{for }  J=
    \left(\begin{smallmatrix}
        0 & -1\\
        1 & 0
    \end{smallmatrix}\right) \:.
\end{equation}
Since $J^2=-\id$, this real algebra is actually the algebra of the complex numbers, $\mathcal{A}\cong M_1(\C) = \C$. Beause of that, $\mathcal{A}$ is Abelian, so it is equal to its commutant in $M_2(\R)$, $\mathcal{A}' = \mathcal{A}$. Hence, the commuting pair $(\mathcal A,\mathcal A')$ cannot be realised as two independent tensor factors on the same real Hilbert space $\R^2$: there is no decomposition of $\R^2$ into two real Hilbert spaces $\R^2\simeq\mathcal H_A^\R\otimes\mathcal H_B^\R$ over which $\mathcal{A}$ and $\mathcal{A'}$ can act independently as $\mathcal{A}\simeq \operatorname{span}_\R\{X_A \otimes\id_B\}$ and $\mathcal A'\simeq\operatorname{span}_\R\{\id\otimes Y_B\}$, as the dimension is simply too low for both factors to be faithfully represented.

After complexification, however, one has $\mathcal A_{\C} = \operatorname{span}_{\C}\{\id,J\}\cong \C\oplus\C$, corresponding to the two eigenspaces of $J$ with eigenvalues $\pm i$. In this complex representation, the two sectors are one-dimensional and can be regarded as trivial tensor products,
$\C\oplus\C \cong (\C\otimes 1)\oplus(1\otimes\C)$. Hence, both can be represented over $\C^2$.

We therefore see that the real algebra is impossible to decompose within $M_2(\R)$ as sector-wise tensor products. This is a different situation than with the corresponding complex algebra, which automatically admits such a decomposition within $M_2(\C)$ in accordance with \cref{thm:AW}. This example illustrates why the complex Artin-Wedderburn argument cannot simply be transplanted to arbitrary real $C^*$-algebras while keeping the same real Hilbert-space representation: in order to use it, as we did in our proof, one must increase the total dimension of the underlying space.
\end{remark}

\begin{remark}[The proof of \cref{prop:real-tsirelson} gives a constructive proof of $\CRtR = \CcR$]
The construction in the proof of \cref{prop:real-tsirelson} also gives an $R$-tensor realisation. To see this, order the direct sums of \cref{eq:Hhat} as
\begin{equation}
    \widehat{\mathcal H}_A
    \cong
    \mathbb R^2\otimes
       \left(\bigoplus_k\mathbb R^{n_k}\right),
    \qquad
    \widehat{\mathcal H}_B
    \cong
    \mathbb R^2\otimes
       \left(\bigoplus_k\mathbb R^{m_k}\right).
\end{equation}
With this ordering, the $R$-product projector for $\widehat{\mathcal H}_A$ and $\widehat{\mathcal H}_B$ is block-diagonal with respect to the sectorization. Its restriction to the matching sector (i.e. when it forces the operators to be non-zero only if they act on the same sector $k$ in $A$ and $B$) is exactly the $R$-product projector $P_{AB}$ for $\mathbb R^{2n_k}$ and $\mathbb R^{2m_k}$. Since $\widehat\sigma$ is supported only on the matching sectors, we have
\begin{equation}
    \TrX{}{\widehat\sigma \left(\widehat A_a^x\otimes_R\widehat B_b^y\right)} = \sum_{k:q_k>0}q_k
    \TrX{}{\sigma_k \left(\Tilde{A}_{a,k}^x\otimes_R \Tilde{B}_{b,k}^y\right)} \:.
\end{equation}
And because 
\[
\sum_{k:q_k>0}q_k \TrX{}{\sigma_k \left(\Tilde{A}_{a,k}^x\otimes_R \Tilde{B}_{b,k}^y\right)} =  \sum_{k:q_k>0}q_k\TrX{}{\sigma_k \left(\Tilde{A}_{a,k}^x\otimes \Tilde{B}_{b,k}^y\right)} \:,
\]
it follows that $\TrX{}{\widehat\sigma \left(\widehat A_a^x\otimes_R\widehat B_b^y\right)}$ is an exact finite-dimensional $R$-tensor real-quantum realisation of $\TrX{}{\sigma \: A_a^xB_b^y} $. Interestingly, the projector $P_{AB}$ used to simulate the complex Kronecker product has the same interpretation as the projectors used to match sectors on each side of the partition.
\end{remark}

\subsection{Equivalence between real and complex quantum correlations}
We have shown that, within both complex and RQT, the different mathematical notions of separated measurements yield the same set of correlations. We now complete the picture by observing that the real and complex theories themselves generate the same set of finite-dimensional correlations.
\begin{theorem}\label{thm:full_equivalence}
    For any finite sets of settings and outcomes $X,Y,A,B$, if the Hilbert-space dimension is finite but unbounded, then all five notions of separated quantum measurements considered above give rise to the same set of correlations. In symbols,
    \begin{equation}
    \CRtR = \CtR = \CcR = \Ct = \Cc:.
    \end{equation}
\end{theorem}
\begin{proof}
    By \cref{thm:equivalence}, the complex commuting and tensor-product models coincide, while \cref{prop:real-tsirelson} establishes the corresponding equivalence for RQT. It therefore remains only to show that the real and complex models generate the same correlations. This was established in Ref.~\cite{HoffreumonWoods2026} by using the $R$-product model as a bridge between $\Ct$ and $\CtR$. There, it was shown that $\Ct = \CRtR = \CtR$. We restated why the equivalence holds in the notation of this note; see \cref{eq:Ct_is_CRtR,eq:Ct_is_CtR} and the discussion around these.
\end{proof}
As a consequence, for any finite number of parties performing measurements with finitely many settings and outcomes on a finite-dimensional shared quantum system, the resulting joint distributions are insensitive both to the choice of mathematical representation of separated measurements and to whether the underlying quantum formalism is formulated over $\R$ or $\C$. Thus, at the level of observable correlations, neither the mathematical representation of separation nor the choice between real and (complex) QT affects the predictions.

\section{RQTs for superselected QTs}
\label{sec:real-superselection-theories}

In this section, the problem of constructing representations of quantum theories with superselection rules within real quantum theory is addressed. We first consider a single Hilbert space, then move to multipartite tensor-product spaces and protocols. We conclude by applying the developed formalism to fermionic information theory. 

The equivalence between the commuting-subsystems and tensor-product-subsystems paradigms for RQT established in the previous section allows us to extend the construction from a single Hilbert space to multipartite tensor-product spaces: each individual (i.e., local) superselected mode is traditionally assumed to be associated with a subset of the observables that commute with the rest. Since we have shown that this picture is equivalent to associating each mode with a tensor factor, we can therefore formulate the construction in the multipartite setting.

It is in this latter picture that the methods we developed in our previous article~\cite{HoffreumonWoods2026} come into play: while each local mode is associated with a tensor factor in quantum theory, it is not guaranteed that its representation over the reals will remain confined to a single tensor factor. In other words, it is not guaranteed that the complex Kronecker product may be replaced by the real Kronecker product when moving the representation from the complex to the real numbers. Our realification procedure ensures exactly that: it is locality-preserving in the sense that tensor factors keep the same meaning between the two pictures\footnote{In technical terms, our realification is a monoidal functor between suitable categories of quantum theories and real quantum theories.}. Hence, the technical contribution of this section is to show that our realification procedure preserves a superselection rule without sacrificing locality.

We will reintroduce some notational conventions already used in the previous section to ease the reading of this part. In the same spirit, the realification conventions are (re)introduced when needed for the construction of the real representation.

\subsection{Superselection over a single Hilbert space}

A superselected quantum theory is simply a quantum theory made invariant under the action of some group $G$. To differentiate between the superselected theory and the finite-dimensional quantum theory it is obtained from, we respectively refer to them as the \textit{invariant} and \textit{unrestricted} (quantum) theories. 
We start by setting up some terminology, which we mainly draw from Ref.~\cite{KMP} and expressed it in the language of operational probabilistic theories (see e.g., Ref.~\cite{Chiribella_2011}).

In quantum theory, to every system $A$ is associated a finite-dimensional complex Hilbert space $\mathcal H_A$. We write $\cB(\mathcal{H}_A)$ for the algebra of all complex linear operators on $\mathcal H_A$, and $\Herm(\mathcal H_A):=\{X\in\cB(\mathcal H_A):X=X^\dagger\}$ for its Hermitian part. 

Groups may act on the operators associated with a system $A$. Let $G$ be a (Lie or finite) group, and let $U_A:G\to\mathrm{U}(\mathcal H_A)$ be a unitary representation of $G$ on $\mathcal H_A$. This induces an adjoint action of $G$ on $\cB(\mathcal H_A)$, defined by
\begin{equation}\label{eq:adjointaction}
	\mathcal U_A(g)(X):=U_A(g)XU_A(g)^\dagger\:.
\end{equation}
An operator $X$ satisfying $\mathcal U_A(g)(X)=X$ for every $g\in G$ is said to be invariant under the action of $G$, or simply invariant. A linear map $\Phi:\cB(\mathcal H_A)\rightarrow\cB(\mathcal H_B)$ is said to be equivariant under the action of $G$ if
\begin{equation}
	\Phi\circ\mathcal U_A(g)=\mathcal U_B(g)\circ\Phi
\end{equation}
for every $g\in G$.

\subsubsection{Invariant and unrestricted QTs}
\label{sec:worlds}

\begin{definition}[Unrestricted and invariant quantum theories]\label{def:worlds}
	For every system $A$, we call the \emph{unrestricted world}, denoted $\mathsf{QT}$, the usual quantum-mechanical quantities associated with it. That is, every density operator, positive-operator-value measure (POVMs), and completely positive trace-preserving maps that can be defined on its associated Hilbert space.
	
	The \emph{invariant world}, denoted $\mathsf{QT}^{\mathrm{inv}}_G$, is the subtheory invariant under the action of $G$ meaning that, for every $g\in G$, its:
	\begin{subequations}
		\begin{align}
			\text{states satisfy}\quad&
			\mathcal U_A(g)(\rho)=\rho \:,
			\label{eq:invworldstate}\\
			\text{effects satisfy}\quad&
			\mathcal U_A(g)(E)=E \:,
			\label{eq:invworldeffect}\\
			\text{channels satisfy}\quad&
			\Phi\circ\mathcal U_A(g)
			=\mathcal U_B(g)\circ\Phi\:.
			\label{eq:invworldchannel}
		\end{align}
	\end{subequations}
	For a quantum instrument $\{\Phi_a\}_a$, condition
	\eqref{eq:invworldchannel} is imposed on every outcome subchannel. For a composite system, a factorisation $U_{AB}(g)=U_A(g)\otimes U_B(g)$ is imposed only when the symmetry is assumed to act on-site relative to that particular bipartition; it is not part of the single-system definition.
\end{definition}

To give more explicit meaning to the definition, $\mathsf{QT}$ and $\mathsf{QT}^{\mathrm{inv}}_G$ denote operational probabilistic theories, namely families of sets indexed by systems and by input--output system pairs. For $\mathsf{QT}$, its state, effect, and channel sets are
\begin{subequations}
	\begin{align}
		\operatorname{St}_{\mathsf{QT}}(A)
		&:=
		\{\rho\in\cB(\mathcal H_A):
		\rho\geq0,\ \TrX{}{\rho}=1\},
		\label{eq:QT-state-set}\\
		\operatorname{Eff}_{\mathsf{QT}}(A)
		&:=
		\{E\in\cB(\mathcal H_A):
		0\leq E\leq\id_{\mathcal H_A}\},
		\label{eq:QT-effect-set}\\
		\operatorname{Chan}_{\mathsf{QT}}(A\!\to\!B)
		&:=
		\left\{\Phi:\cB(\mathcal H_A)\to\cB(\mathcal H_B):
		\Phi\ \text{is completely positive and trace preserving}\right\}.
		\label{eq:QT-channel-set}
	\end{align}    
\end{subequations}
And the pairing between states and effects, $\operatorname{St}_{\mathsf{QT}}(A) \times \operatorname{Eff}_{\mathsf{QT}}(A) \to [0,1]$, is the usual Born rule, i.e. the Frobenius inner product $(\rho,E)\mapsto \TrX{}{\rho E}$. 

In addition, we define its set of \textit{measurements} given a fixed and finite outcome set $\Omega$ as
\begin{equation}
	\operatorname{POVM}_{\mathsf{QT}}(A;\Omega)
	=
	\left\{\{E_a\}_{a\in\Omega}:
	E_a\in\operatorname{Eff}_{\mathsf{QT}}(A),\
	\sum_{a\in\Omega}E_a=\id_{\mathcal H_A}\right\},
	\label{eq:QT-POVM-set}
\end{equation}
We denote and call this set the `set of POVMs' since that is what they are for both theories. 
Also for a (finite) outcome set $\Omega$, we define a $\mathsf{QT}$ instrument from $A$ to $B$ as a family $(\Phi_a)_{a\in\Omega}$ of completely positive, trace-nonincreasing maps whose sum $\sum_{a\in\Omega}\Phi_a$ belongs to $\operatorname{Chan}_{\mathsf{QT}}(A\to B)$.

Similarly, the invariant theory is the operational probabilistic theory obtained from the following subsets of the unrestricted sets:
\begin{subequations}\label{eq:QTinv}
	\begin{align}
		\operatorname{St}_{\mathsf{QT}^{\mathrm{inv}}_G}(A)
		&:=
		\{\rho\in\operatorname{St}_{\mathsf{QT}}(A):
		\mathcal U_A(g)(\rho)=\rho\ \ \forall g\in G\},
		\label{eq:QTinv-state-set}\\
		\operatorname{Eff}_{\mathsf{QT}^{\mathrm{inv}}_G}(A)
		&:=
		\{E\in\operatorname{Eff}_{\mathsf{QT}}(A):
		\mathcal U_A(g)(E)=E\ \ \forall g\in G\},
		\label{eq:QTinv-effect-set}\\
		\operatorname{Chan}_{\mathsf{QT}^{\mathrm{inv}}_G}(A\!\to\!B)
		&:=
		\left\{\Phi\in\operatorname{Chan}_{\mathsf{QT}}(A\to B):
		\Phi\circ\mathcal U_A(g)
		=\mathcal U_B(g)\circ\Phi\ \ \forall g\in G\right\}.
		\label{eq:QTinv-channel-set}
	\end{align}
\end{subequations}
In accordance, its set of POVMs (for a given $\Omega$) is 
\begin{equation}\label{eq:QTinv-POVM-set}
	\operatorname{POVM}_{\mathsf{QT}^{\mathrm{inv}}_G}(A;\Omega) = \left\{(E_a)_{a\in\Omega}\in \operatorname{POVM}_{\mathsf{QT}}(A;\Omega): E_a\in\operatorname{Eff}_{\mathsf{QT}^{\mathrm{inv}}_G}(A)\ \forall a\in\Omega\right\} \:,
\end{equation}
An invariant instrument is defined analogously: every outcome subchannel is equivariant, and their sum is an element of $\operatorname{Chan}_{\mathsf{QT}^{\mathrm{inv}}_G}(A\to B)$. Both theories have the same system labels and Hilbert spaces; the sequential composition is ordinary composition of maps, and the parallel composition is the ordinary Hilbert-space tensor product. The tensor unit, or trivial input system, is the system $0$ with $\mathcal H_0=\C$ and $U_0(g)=\id_1 = 1$ for every $g\in G$. 

\begin{remark}[Relation to KMP's terminology] \label{rem:KMPworld}
	Kitaev, Mayers, and Preskill (KMP) define their invariant world by requiring physical operations to commute with the group action \cite[Sec.~II.B--C]{KMP}. For states, \cref{eq:invworldstate} is equivalent to $[U_A(g),\rho]=0$.  \cref{eq:invworldeffect,eq:invworldchannel} are the effect and channel versions of that requirement.  
	
	Here, the \emph{trivial representation} means precisely this one-dimensional input system $\mathcal H_0=\C$ on which every group element acts as the identity. Treating a state preparation $\rho$ on $A$ as the $(0 \to A)$ channel, $\Phi_\rho:\cB(\mathcal H_0)\to\cB(\mathcal H_A)$ such that $\Phi_\rho(1)=\rho$, the channel equivariance condition becomes that, for every $g\in G$:
	\[ \quad \Phi_\rho\circ\mathcal U_0(g)
	=\mathcal U_A(g)\circ\Phi_\rho
	\quad\iff\quad
	\mathcal U_A(g)(\rho)=\rho 
	\quad\iff\quad
	[U_A(g),\rho] = 0
	\:.
	\]
	This is exactly \eqref{eq:invworldstate}. The following discussion is unchanged if one elects to keep individual charge sectors as separate system types, because it uses only set inclusion.
\end{remark}

\paragraph{Forgetting the superselection restriction.}
The displayed definitions make the relation between the two worlds literal: each invariant state, effect, POVM, channel, and instrument is already an element of the corresponding unrestricted set.  We therefore write $\mathsf{QT}^{\mathrm{inv}}_G\subseteq\mathsf{QT}$ and represent an invariant protocol in unrestricted quantum theory by leaving every Hilbert space, operator, channel, instrument outcome, and classical control label unchanged.

Because nothing is replaced, every Born probability, every unnormalised and normalised post-measurement state, and every joint or conditional probability in a finite adaptive protocol is identical in the two descriptions.
Likewise, if a channel, instrument, or measurement is assigned a particular ordinary tensor-product block decomposition, the same decomposition remains true after the restriction is forgotten. This statement concerns the locality of operations; it makes no additional assertion about the factorisation or separability of the initial state.

Therefore, no reference frame is needed for our purposes. If a reference system $R$ is adjoined anyway, it may remain a spectator.  In particular, KMP \cite[Eq.~(30)]{KMP} show that their reference-frame image of an operator already commuting with the symmetry is simply $\id_R\otimes M$.

\subsubsection{Real representation of QT and unitary invariance}
\label{sec:embedded-real-theory}
For $A$ a $d\times d$ complex matrix, i.e. for $A\in M_d(\C)$, the standard complex to real mapping $\Gamma_d : M_d(\C) \to M_{2d}(\R)$ has been defined in the introduction under the following convention(s) of notation:
\begin{equation}
	\Gmap_d(A):= I_d \otimes \mathrm{Re}(A) + J_d \otimes \mathrm{Im}(A) = 
	\begin{pmatrix}
		\mathrm{Re}(A)&-\mathrm{Im}(A)\\
		\mathrm{Im}(A)& \mathrm{Re}(A)
	\end{pmatrix}\:,
	\label{eq:Gamma}
\end{equation}
with $I = \left(\begin{smallmatrix}
	1 & 0 \\ 0 & 1 
\end{smallmatrix} \right)$, $J = \left(\begin{smallmatrix}
	0 & -1 \\ 1 & 0 
\end{smallmatrix} \right)$, and $I_d := I \otimes \id_d$, $J_d := J \otimes \id_d$.

It should be clear that $\Gmap$ is real-linear and injective. We normally suppress the dimension subscript and write $\Gmap$ when the domain is clear. Direct block multiplication shows that it is a $C^*$-algebra homomorphism that preserves positivity and the trace up to a factor of 2 (see e.g. Ref.~\cite{Hoffreumon2025} for proof):
\begin{subequations}
	\begin{align}
		\Gmap(AB)&=\Gmap(A)\Gmap(B)\:,&
		\Gmap(A^\dagger)&=\Gmap(A)^T,
		\label{eq:starhom}\\
		A\geq0&\Longleftrightarrow\Gmap(A)\geq0
		\quad(\Gmap(A)=\Gmap(A)^T),&
		\Gmap(\id_d)&=\id_{2d}\:,
		\label{eq:positiveunital}\\
		\TrX{}{\Gmap(A)}&=2\mathrm{Re}(\TrX{}{A}) \:,&
		\TrX{}{\Gmap(A)^T\Gmap(B)}
		&=2\mathrm{Re}(\TrX{}{A^\dagger B}).
		\label{eq:traceisometry}
	\end{align}
\end{subequations}
In every Hermitian or Born-pairing application below, the complex trace on the right is real, so the real-part symbol may be omitted. This implies that the Frobenius inner product is preserved between self-adjoint elements, so the embedding is an isometry from the real inner product space of complex Hermitian matrices to the one of real symmetric matrices.

Recall that real quantum theory (RQT) on a finite-dimensional real Hilbert space uses real symmetric positive-semidefinite trace-one matrices as states, real symmetric effects verifying $0\leq F\leq\id$, POVMs whose effects sum to the identity, and real completely positive trace-preserving maps as transformations. Hence the $\Gmap$ embeds $\mathsf{QT}$ into $\mathsf{RQT}$. 

In addition, the realification map $\Gmap$ sends the pure transformations of $\mathsf{QT}$, i.e., the unitary maps, to a subset of the pure transformations of $\mathsf{RQT}$, i.e., the orthogonal maps. Indeed, let $\mathrm U(d)$ denote the complex unitary matrices on $\C^d$ and $SU(d)$ its special subset of determinant-one elements, let $\mathrm{O}(2d)$ and $SO(2d)$ similarly denote the real (special) orthogonal $d\times d$ matrices, then the following holds
\begin{equation}
	\forall U \in U(d)\:: \qquad \Gmap(U) \in SO(d) \text{ and } [J_d,\Gmap(U)] = 0\:.
\end{equation}
This quickly follow from
\begin{equation}
	O(g)^TO(g) = \Gmap(U(g)^\dagger U(g))=\id_{2d}\:,
	\label{eq:single-orthogonal}
\end{equation}
and 
\begin{equation}
	O(g)O(h) =\Gmap(U(g)U(h))=\Gmap(U(gh))=O(gh) \:.
	\label{eq:single-representation}
\end{equation}
To verify the determinant, complexify the real matrix $\Gmap(A)$: the change of variables $(x,y)\mapsto(x+iy,x-iy)$ conjugates it to $\operatorname{diag}(A,\overline A)$.  Hence $\det\Gmap(A)=\det A\det \overline A =|\det_{\C}A|^2$, and therefore $\det O(g)=1$, proving that $O(g)\in\mathrm{SO}(2d)$.

Hence, a special unitary is mapped to a special orthogonal matrix, but so is any unitary.

\paragraph{Matrix/channel convention.}
Whenever a reversible symmetry matrix and its conjugation channel are both used, a plain capital letter denotes the matrix, and the corresponding calligraphic letter denotes the channel.  Thus $U,V,O,J$ are matrices, while $\mathcal U,\mathcal V,\mathcal O,\mathcal J$ act on operators.  For parity we retain the standard notation $\Theta$ (and $\widetilde\Theta$) for the conjugation automorphism implemented by the matrix $P$ (and
$\widetilde P$).

We now impose an arbitrary physical superselection on the unrestricted QT via a group $G$. This construction uses only the commutation condition $[X,J_d]=0$ and the unitary representation $U : G \to U(d)$; its later interpretation as an auxiliary real superselection rule via the orthogonal representation $O : G \to O(2d)$ is not needed here. That is, although the real representation of a superselection via $\Gmap$ is automatically a real superselection on the image $\Gmap(\mathsf{QT})$ of quantum theory within real quantum theory, we will not try to extend it into a real superselection rule on all of RQT.

\begin{definition}[Fixed-point algebras]
	Let $U:G\to\mathrm U(d)$ be a finite-dimensional unitary representation, define
	\begin{equation}
		O(g):=\Gmap(U(g)) \:.
		\label{eq:global-single-O}
	\end{equation}
	Then $O:G\to\mathrm{SO}(2d)$ is its associated orthogonal representation. Let the corresponding conjugation	channels be respectively denoted by
	\begin{align}\label{eq:single-system-conjugation-channels}
		\mathcal U(g)(A)&:=U(g)AU(g)^\dagger,&
		\mathcal O(g)(X)&:=O(g)XO(g)^T\:,
	\end{align}
	in accordance with the matrix/channel convention. Then the \emph{complex and real fixed-point algebras of $G$} are respectively defined as the sets:
	\begin{subequations}
		\begin{align}
			\mathfrak A_{\C}^{G}
			&:=\{X\in M_d(\C):\mathcal U(g)(X)=X\ \forall g\in G\},
			\label{eq:complexfixed}\\
			\mathfrak A_{\R}^{G,J}
			&:=\{Y\in M_{2d}(\R):
			\mathcal O(g)(Y)=Y\ \forall g\in G,\ \mathcal{J}_d(Y) = Y \}\:.	
			\label{eq:realfixed}
		\end{align}
	\end{subequations}
\end{definition}
The second requirement in definition \cref{eq:realfixed}, $[Y,J_d]=0$, selects the real matrices that represent operators of quantum theory.  Perhaps unsurprisingly with respect to the discussion so far, what we call the real fixed-point algebra is not exactly the set of operators invariant under the orthogonal representation of $G$. Rather, it is the realified complex fixed-point algebra.

In the definitions \eqref{eq:complexfixed} and \eqref{eq:realfixed}, the first requirements are respectively equivalent to $[X,U(g)]=0$ and $[Y,O(g)]=0$ for every $g$. We therefore see that the fixed-point algebras are a quick way to refer to the set of all operators invariant with respect to a given action of a group, so imposing a superselection is equivalent to restricting the operators of the theory to the fixed-point algebra of $G$. 

Indeed, because $\cB(\mathcal{H}_A) \cong M_{d_A}(\C)$, we can define the fixed point algebra of $G$ with respect to its action on a system $A$, 
\begin{equation}
	\mathfrak A_{\C}^{G}(A):=\{X\in \cB(\mathcal{H}_A):\mathcal U_A(g)(X)=X\ \forall g\in G\}\:.
\end{equation}
Hence, the invariant QT can be expressed as the QT over the fixed-point algebra, e.g., $\operatorname{St}_{\mathsf{QT}^{\mathrm{inv}}_G}(A) = \operatorname{St}_{\mathsf{QT}}(A)\cap \mathfrak A_{\C}^{G}(A)$. It remains to express the transformations between systems  (channels) using fixed-point algebras. As the equivariant channels are but structure-preserving mappings between sets of invariant operators, one can define the corresponding sets of equivariant channels as
\begin{subequations}
	\begin{align}
		&\mathsf{Hom}_\C^G(A\to B) := \{\Phi\in M_{d_A}(\C) \to M_{d_B}(\C): \mathcal{U}_B(g) \circ \Phi \circ \mathcal{U}_A^\dag(g) = \Phi\ \forall g\in G\}\:, \\
		&\mathsf{Hom}_\R^{G,J}(A\to B) \notag \\
		&:= \{\widetilde{\Phi}\in M_{2d_A}(\R) \to M_{2d_B}(\R): \mathcal{O}_B(g) \circ \widetilde{\Phi} \circ \mathcal{O}_A^T(g) = \widetilde{\Phi}\ \forall g\in G, \mathcal{J}_{d_B} \circ \widetilde{\Phi} \circ \mathcal{J}_{d_A}^T=\widetilde{\Phi}\}.
	\end{align}
\end{subequations}

Under these definitions, we see that the invariant world can be described as:
\begin{subequations}
	\begin{align}
		\operatorname{St}_{\mathsf{QT}^{\mathrm{inv}}_G}(A)
		&=
		\{\rho\in \mathfrak A_{\C}^{G}(A): \: \rho \geq 0, \TrX{}{\rho} = 1\} \:, \\
		\operatorname{Eff}_{\mathsf{QT}^{\mathrm{inv}}_G}(A)
		&=
		\{ E\in \mathfrak A_{\C}^{G}(A):\: 0 \leq E \leq \id_{d_A} \} \:, \\
		\operatorname{Chan}_{\mathsf{QT}^{\mathrm{inv}}_G}(A\!\to\!B)
		&=
		\{\Phi\in\mathsf{Hom}_\C^G(A\to B) : \: \Phi \:\mathrm{CPTP} \} \:.
	\end{align}
\end{subequations}

We are now ready to present the real representation.
\begin{lemma}\label{lemm:fixediso}
	The restriction
	\begin{equation}
		\Gmap:\mathfrak A_{\C}^{G}\longrightarrow
		\mathfrak A_{\R}^{G,J}
		\label{eq:fixediso}
	\end{equation}
	is a bijective real $*$-algebra isomorphism, where the involutions are complex adjoint on the source and real transpose on the target.  
\end{lemma}
\begin{proof}
	The realification map is a real $*$-algebra homomorphism that maps $U(d)$ to $SO(2d)$, it remains to show that the two algebras are isomorphic. Thus, proving the condition $\mathcal{U}(g)(X) = X$ equivalent to $\mathcal{O}(g)(\Gmap(X)) = \Gmap(X)$ and $[\Gamma(X),J_d] = 0]$.
	
	For the second condition, commutation of the matrices in the image of $\Gmap$ with $J_d$ directly follows from the definition  \eqref{eq:Gamma}, as mentioned earlier. Conversely, write a real $2d\times2d$ matrix in $d\times d$ blocks.
	Commutation with
	$J_d=\left(\begin{smallmatrix}0&-\id_d\\
		\id_d&0\end{smallmatrix}\right)$ is equivalent to the existence of real
	matrices $B,C$ such that
	\begin{equation}
		X=\begin{pmatrix}B&-C\\ C&B\end{pmatrix}
		=\Gmap(B+iC).
		\label{eq:complex-structure-centralizer}
	\end{equation}
	Hence every $X\in\mathfrak A_{\R}^{G,J}$ has a unique preimage
	$A=B+iC$.  Equation \eqref{eq:commutator} and injectivity then imply
	$[A,U(g)]=0$ for every $g$, proving surjectivity of \eqref{eq:fixediso}.
	
	Now, for every $X$ and $g$,
	\begin{equation}
		[\Gmap(X),O(g)]
		=[\Gmap(X),\Gmap(U(g))]
		=\Gmap([X,U(g)]) \:.
		\label{eq:commutator}
	\end{equation}
	
	Since $\mathcal{U}(g)(X) = X \iff [U(g),X]=0$ the last term is equal to $\Gmap(0 \id_d)$ which is 0 by real-linearity, hence $[\Gmap(X),O(g)] = 0$ and so the first condition follows, $\mathcal{O}(g)(\Gmap(X)) = \Gmap(X)$.
\end{proof}

\subsubsection{Embedding the superselected QT}
We now have established enough conventions and properties to directly establish what is the real representation of a superselected quantum theory.
\begin{theorem}[The embedded real quantum theory of the symmetry]
	\label{thm:algebra}
	The following triple of sets:
	\begin{align}
		\mathsf S_{\R}^{G,J} (A)
		&:=\{\sigma\in\mathfrak A_{\R}^{G,J}(A): \sigma^T=\sigma,\ \sigma\geq0,\ \TrX{}{\sigma}=1\} \:,
		\label{eq:real-fixed-state-set}\\
		\mathsf E_{\R}^{G,J}(A)
		&:=\{F\in\mathfrak A_{\R}^{G,J}(A): F^T=F,\ 0\leq F\leq\id_{2d}\} \:,
		\label{eq:real-fixed-effect-set}\\
		\mathsf C_{\R}^{G,J}(A\to B) 
		&:= \{ \widetilde{\Phi} \in \mathsf{Hom}_\R^{G,J}(A\to B): \: \widetilde{\Phi} \: \mathrm{CPTP} \} \:, \label{eq:real-fixed-channel-set}
	\end{align}
	are isomorphic to $(\operatorname{St}_{\mathsf{QT}^{\mathrm{inv}}_G}(A),\: \operatorname{Eff}_{\mathsf{QT}^{\mathrm{inv}}_G}(B) ,\: \operatorname{Chan}_{\mathsf{QT}^{\mathrm{inv}}_G}(A \to B) )$.
	In particular, the mapping
	\begin{subequations}
		\begin{align}
			(\frac{1}{2}\Gmap, \Gmap, \mathrm{Ad}_\Gmap): & \operatorname{St}_{\mathsf{QT}^{\mathrm{inv}}_G}(A) \times \operatorname{Eff}_{\mathsf{QT}^{\mathrm{inv}}_G}(A) \times \operatorname{Chan}_{\mathsf{QT}^{\mathrm{inv}}_G}(A \to B) \longrightarrow \mathsf S_{\R}^{G,J} \times  \mathsf E_{\R}^{G,J} \times \mathsf{C}_\R^{G,J} :: \notag \\
			\rho&\longmapsto\tfrac12\Gmap(\rho) \:; \label{eq:fixed-state-bijection}\\
			F&\longmapsto\Gmap(F) \:;\label{eq:fixed-effect-bijection} \\
			\Phi &\longmapsto \Gmap \circ \Phi \circ \Gmap^{-1} := \mathrm{Ad}_\Gmap(\Phi)\:. \label{eq:fixed-channel-bijection}
		\end{align}
	\end{subequations}
	yield an isometry, meaning that they preserve every Born pairing:
	\begin{equation}
		\TrX{}{E \Phi(\rho) }  = \TrX{}{\Gmap(E) \mathrm{Ad}_\Gmap(\phi) \circ\tfrac12\Gmap(\rho)}\:.
		\label{eq:fixedborn}
	\end{equation}
	Consequently the $G$-invariant state--effect subtheory of quantum theory is
	operationally isomorphic to the explicitly defined subtheory of RQT
	\cref{eq:real-fixed-state-set,eq:real-fixed-effect-set,eq:real-fixed-channel-set}, which
	carries the analogous orthogonal action $\mathcal O$.
\end{theorem}

\begin{proof}
	If $\rho$ is a state of $A$, then $\tfrac12\Gmap(\rho)$ is symmetric,
	positive, invariant under $\mathcal O(g)$ for every $g\in G$, and commutes with $J_{d_A}$; furthermore,
	\[
	\TrX{}{\tfrac12\Gmap(\rho)} =\TrX{}{\rho}=1.
	\]
	Thus the state map lands in $\mathsf S_{\R}^{G,J}$.  Conversely, let
	$\sigma\in\mathsf S_{\R}^{G,J}$.  By \cref{lemm:fixediso},  \cref{eq:fixediso}
	gives a unique $X\in\mathfrak A_{\C}^{G}$ with $\sigma=\Gmap(X)$.
	Because $\sigma^T=\sigma$, adjoint preservation and injectivity give $X^\dagger=X$; positivity reflection gives $X\geq0$; and the trace
	identity gives
	\[
	1=\TrX{}{\sigma}=2\TrX{}{X}\:.
	\]
	Consequently $\rho:=2X$ lies is a state of $A$ and $\sigma=\tfrac12\Gmap(\rho)$.  Injectivity makes this preimage unique.
	
	The effect argument is the same without trace rescaling.  If  $F\in\mathsf E_{\R}^{G,J}$, write uniquely $F=\Gmap(E)$.  Symmetry gives $E^\dagger=E$, while positivity reflection applied to $F$ and $\id_{2{d_B}}-F=\Gmap(\id_{d_B}-E)$ gives $0\leq E\leq\id_{d_B}$.  The converse follows directly from positivity, unitality, and \eqref{eq:commutator}.  This proves both bijections.
	Finally, \eqref{eq:fixedborn} is the Hermitian case of \eqref{eq:traceisometry}; its factor $1/2$ exactly cancels the doubled real trace.
	
	The proof for the channels follows by construction, and the preservation of the Born rule comes from \eqref{eq:traceisometry} again.
\end{proof}
\begin{remark}[Projective and antiunitary symmetries]
	The theorem assumes a genuine unitary representation.  If it were projective instead, i.e. $U(g)U(h)=\omega(g,h)U(gh)$ for a norm-1 scalar $\omega(g,h) \in \C$ associated with every pair $g,h$, then $O(g)O(h)=\Gmap(\omega(g,h)\id_d)O(gh)$. Here, $\Gmap(\omega(g,h)\id_d)$ is the two-dimensional phase rotation repeated $d$ times; it does not correspond to a projective orthogonal representation which would have required a real norm-1 scalar $\tilde{\omega}(g,h)=\pm 1$ repeated $2d$ times. Hence, in order to properly map a projective representation of $G$ over the complex to the reals, the naive approach via $\Gmap$ appears to fail. Nevertheless, the associated central extension of $G$, in this case $U(1)$, is properly mapped over to the real, i.e., to $SO(2)$. 
	
	Similarly, antiunitary symmetries seem to be impossible to map to the reals since $\R$ has no non-trivial notion of antilinearity. This naive analysis is once again limited, as the antiunitary maps are precisely mapped to the part of $SO(n)$ where the unitary maps are not mapped. That is, to the special orthogonal matrices that anticommute with $J_d$ \cite{Hoffreumon2025}.
	
	We leave these two very interesting cases for future work. Here we are solely concerned with superselection rules seen as a unitary representation of a group.
\end{remark}

\subsubsection{Interpreting the embedded superselected QT as a superselected RQT}
In \cref{thm:algebra}, the condition $[X,J_d]=0$ is essential. Ordinary RQT also contains real states and effects that do not preserve this encoded complex structure. The theorem therefore identifies the superselected quantum theory with a specific subtheory of RQT; it does not identify it with every $O(G)$-invariant state or effect allowed by RQT. In other words, the $G$-invariant QT is not mapped to a $G$-invariant RQT, but only to a subtheory of it.

One may nevertheless ask whether it is possible to interpret the real representation of an invariant QT as an invariant RQT. To do so, we first need to define what superselection is in RQT. Naturally, we take the same definition as in the complex case: a superselected RQT is an RQT made invariant under the action of some group $G$. What differs is that the action of $G$ over a real state is defined using its orthogonal representation. Notice that this makes superselection rarer in RQT since not every group possesses an orthogonal representation. To ensure the existence of an orthogonal representation, we shall assume $G$ to be compact.

With this definition, it is quite simple to see how invariant subtheories of QT are invariant subtheories of RQT. If we make explicit all conditions in \cref{eq:real-fixed-state-set}, we get:
\begin{equation}
	\begin{aligned}
		\sigma \in \mathsf S_{\mathbb{R}}^{G,J} (A) \iff&  \sigma\in M_{2d_A}(\mathbb{R}): \: \sigma^T=\sigma,\ \sigma\geq 0,\ \TrX{}{\sigma}=1 \:, \\
		& \forall g \in G: \: [\Gmap(U(g)),\sigma] = 0 \:,\\
		& [J_{d_A},\sigma] = 0 \:.
	\end{aligned}
\end{equation}
Now, the first line of conditions simply asks $\sigma$ to be a valid state in the embedding $2d_A \times 2d_A$ RQT. In our notation, this line reads $\sigma \in \operatorname{St}_{\mathsf{RQT}}(A)$.

The second line, $\forall g \in G: \: [\Gmap(U(g)),\sigma] = 0$, looks as though it imposes a superselection with respect to an orthogonal representation $O: G \to M_{2d_A}(\mathbb{R})$. 
However, the presence of the third line, which reads $[J_{d_A},\sigma] = 0$, does not allow one to reach this conclusion from the outset. Nevertheless, this condition also looks like a superselection rule, and indeed it is. The matrix $J_{d_A}$, which we used to encode the imaginary unit since it is the image of $i\id_{d_A}$, implies the orthogonal representation of a discrete group. Indeed, since $i^4 = 1 \implies J_{d_A}^4 = \id_{2d_A}$, we see that $J_{d_A}$ is the generator of the discrete group $C_4 = \langle a | a^4 = e \rangle$. The implied orthogonal representation is 
\begin{equation}
	O_{C_A}: C_4 \to M_{2d_A}(\mathbb{R}), \quad a \mapsto J_{d_A} \:,
\end{equation}
such that the four elements of the group are given by
\begin{equation}
	O_{C_A}(e) = \id_{2d_A}, \quad O_{C_A}(a) = J_{d_A}, \quad O_{C_A}(a^2) = -\id_{2d_A}, \quad O_{C_A}(a^3) = -J_{d_A}\:,
\end{equation}
and it is straightforward to check that this is an orthogonal representation since $O_{C_A}(g)^T O_{C_A}(g) = \id_{2d_A}$ for every $g \in C_4$.

So we have two groups present, $G$ and $C_4$, and the states are invariant under both. It is evident that the states are actually invariant under $C_4 \times G$, but in order to conclude this, one must first ensure that the elements of both groups commute pairwise. By construction, we know this is the case: $[\Gmap(U(g)), O_{C_4}(h)] = 0$ for all $g\in G$ and for all $h\in C_4$ is trivially implied by $[\Gmap(U(g)), J_{d_A}] = 0$ for all $g$. This relation simply encodes that the matrices $U(g)$ commute with complex numbers in the preimage. From this, we can conclude that the states are invariant under the orthogonal action of the group $C_4 \times G$. The reasoning extends directly to effects and channels.

Consequently, any unitary representation of $G$ corresponds to an orthogonal representation of $C_4 \times G$. What we have established is that a QT in $d \times d$ dimensions is isomorphic to a $C_4$-superselected RQT, and that any superselected subtheory of QT is isomorphic to a superselected subtheory of RQT. In symbols:
\begin{subequations}
	\begin{align}
		&\mathsf{QT} \cong \mathsf{RQT}^{\mathrm{inv}}_{C_4} \:;\\
		&\mathsf{QT}^{\mathrm{inv}}_G \cong \mathsf{RQT}^{\mathrm{inv}}_{C_4\times G} \:.
	\end{align}
\end{subequations}

\subsection{Superselection over multiple Hilbert spaces}

We now turn to the multipartite case. The subsystems are initially specified by a collection of pairwise commuting subalgebras of the algebra of observables. As we have shown in \cref{sec:Tsirelson}, in finite dimensions, we can equivalently represent these subsystems as tensor factors, both over the reals and in the presence of superselection constraints. We therefore assume throughout this section that the Hilbert space decomposes as $\mathcal H=\bigotimes_{i=1}^n\mathcal H_i$, with each label $i$ associated with a subsystem.

A representation of a group on a multipartite Hilbert space need not factorise. Here, however, we want the same global symmetry $G$ to act locally on each subsystem. Thus, for the decomposition $\mathcal H=\bigotimes_{i=1}^n\mathcal H_i$, we assume that each factor is invariant under a representation $U_i:G\to \cB(\mathcal{H}_i)$. Furthermore, we assume that the global representation is given by the corresponding diagonal (or on-site) action
\begin{equation}\label{eq:globalU}
	U(g)=\bigotimes_{i=1}^n U_i(g)\:.
\end{equation}
Using the same group element $g$ in every factor ensures that these local actions combine into a global representation of $G$, so that the superselection rule is also implemented globally. The global representation is therefore generally reducible, but this is irrelevant for our purposes. What matters is that the symmetry is simultaneously \emph{global}, in that one common $g\in G$ acts on all subsystems, and \emph{local}, in that its representing operator factorises across the chosen tensor decomposition.

This factorisation is an additional assumption about the symmetry and the chosen partition, rather than a consequence of the tensor-product decomposition itself. If the symmetry mixes modes or laboratories, the corresponding degrees of freedom must instead be grouped into a larger subsystem, or the symmetry must be regarded as nonlocal relative to the chosen partition. 

\subsubsection{Realification of multiple tensor-product subsystems}

We mentioned in the introduction that the standard complex-to-real mapping cannot be used directly over tensor-product spaces because it does not intertwine the complex tensor-product structure with the corresponding real tensor-product structure. In this sense, it is not locality-preserving. This is precisely where our locality-preserving construction comes into play. See Ref.~\cite[Apps.~A and C]{HoffreumonWoods2026} for a detailed presentation.

Here we use the construction in its simplest form. For the $2\times2$ real matrices $I^{(1)}:=I$ and $J^{(1)}:=J$, define the pair of $(2^n\times2^n)$ real matrices $(I^{(n)},J^{(n)})$ recursively by
\begin{subequations}\label{eq:InJn}
	\begin{align}
		I^{(n)}
		&=\tfrac12\left(I^{(n-1)}\otimes I^{(1)}
		-J^{(n-1)}\otimes J^{(1)}\right)\:;\\
		J^{(n)}
		&=\tfrac12\left(J^{(n-1)}\otimes I^{(1)}
		+I^{(n-1)}\otimes J^{(1)}\right) \:.
	\end{align}
\end{subequations}
After permuting the tensor factors so that each auxiliary two-dimensional real quantum system encoding the complex phase of a complex factor is placed to the left of not only that real factor but of all $n$ such factors, i.e., after applying the permutation
\[
\bigotimes_{i=1}^n M_{2d_i}(\R) \cong \left(\bigotimes_{i=1}^n M_{2}(\R)\right) \otimes \left(\bigotimes_{i=1}^n M_{d_i}(\R)\right) \:,
\]
define
\begin{equation}\label{eq:Gamman}
	\begin{aligned}
		\Gmap^{(n)} : \bigotimes_{i=1}^n M_{d_i}(\C)
		&\longrightarrow
		\left(\bigotimes_{i=1}^n M_{2}(\R)\right) \otimes \left(\bigotimes_{i=1}^n M_{d_i}(\R)\right)\:, \\
		X &\longmapsto
		\Gmap^{(n)}(X)  :=I^{(n)}\otimes\mathrm{Re}(X) +J^{(n)}\otimes\mathrm{Im}(X)\:.
	\end{aligned}
\end{equation}

We previously proved that this map can be used to construct a locality-preserving embedding of quantum-theory protocols, viewed as operational probabilistic theories, into real-quantum-theory protocols~\cite[Theorem~2]{HoffreumonWoods2026}. In particular, we showed that every sequence of complex quantum instruments composed in succession (via $\circ$) or in parallel (via $\otimes$) admits a corresponding sequence of real quantum instruments, with states and effects recovered as special cases. To make this precise, we need to introduce some notation for what we describe as a `protocol' in a theory. (See Ref.~\cite{HoffreumonWoods2026} for a more comprehensive and pedagogical description.)

The protocol involves $n$ subsystems $\{A_i\}_{i=1}^n$ that undergo $r$ rounds of interventions, i.e., evolutions described by instruments, before being finally discarded. The systems are assumed to start in an arbitrary state $\rho^{(0)} \in \cB\left(\bigotimes_i \mathcal{H}_{A_i}^{(0)}\right)$, which represents the state of the overall system before the parties involved in the protocol start acting on it.

The protocol is therefore associated with a Born rule of the form $\TrX{}{\left[\Phi^{(r)} \circ \ldots \circ \Phi^{(1)}\right](\rho^{(0)})}$.
Since we are dealing with channels, the Hilbert-space dimension associated with each subsystem may change between rounds of the protocol. In particular, a local dimension $d_{A_i}$ associated with subsystem $A_i$ may be set to $1$ to represent its discarding, and it can be increased to $d_{A_i}>1$ to represent its (re)preparation. For this reason, $d_{A_i}^{(t)}$ may differ from $d_{A_i}^{(t-1)}$ for every $i = 1,\ldots,n$ and every $t=1,\ldots,r$.

There are still two ingredients missing from the description. The first is the specification of how the subsystems are distributed among groups of parties at each round, which we will call a \emph{partition}; the second is the specification of the set of settings distributed to the parties at the beginning of a round, which we will call a \emph{context}.

At round $t$, a partition $\mathsf{P}^{(t)}= \{P_1^{(t)},\ldots,P_{m_t}^{(t)}\}$ dictates an assignment of the $n$ subsystems to $m_t\leq n$ subsets associated with the parties $P_k^{(t)}$. For example, if at round $t$ the first party has access to subsystems $A_1$ and $A_2$, this is represented by $P_1^{(t)}=\{A_1,A_2\}$. Hence, the partition encodes the locality data of a round, which consequently induces a factorisation of the channel. Given $\mathsf{P}^{(t)}$, the operation at round $t$ must have the form
\begin{equation}
	\Phi^{(t)} \cong \Phi^{(t)}_{P_1^{(t)}} \otimes \Phi^{(t)}_{P_2^{(t)}} \otimes ... \otimes \Phi_{P_{m_t}^{(t)}}^{(t)} \:,
\end{equation}
where
\begin{equation*}
	\begin{aligned}
		&\Phi^{(t)} \in \operatorname{Chan}\left((\{A_1,...A_n\} \to \{A_1,...A_n\}\right)\:; \\
		& \Phi^{(t)}_{P_k} \in \operatorname{Chan}\left(P_k^{(t)} \to P_k^{(t)}\right) \:.
	\end{aligned}
\end{equation*}
The set of partitions encodes the \emph{locality structure} of the protocol, in the sense of which subsystems can be jointly accessed at any given round. One can equivalently call this the circuit topology.

It remains to describe how the context, i.e., how settings and previous outcomes are distributed among parties. In a general adaptive protocol, the settings received by one party at round $t$ may depend on both the settings and outcomes of all parties at every round before $t$. This is used to represent the classical communication that may take place between rounds.

However, to simplify the description, and since the proof is essentially unchanged, we avoid these complications by assuming that only the last round actually receives settings and is expected to produce outcomes. No conditional operation is assumed to take place before the last round. In other words, we assume that only the last round is described by a quantum instrument, while all previous rounds are instead described by channels. Again for simplicity, we assume that the number of parties at the end of the protocol is the same as the number of subsystems. This is without loss of generality, as any given subsystem can always be refactored into smaller subsystems.

At the end of the protocol, the parties are assumed to choose their measurement operations using a setting $x_k$ sampled from a finite set $X_k$ of classical measurement settings and, for each $x_k\in X_k$, they obtain an outcome $a_k$ sampled from a finite outcome set $A_k(x_k)$. In accordance with the previous simplifications, the last round of the protocol can be seen as the application of $n$ POVMs $\{E^i_{a_i|x_i}\}_{a_i\in A_i(x_i), x_i \in X_i}$, such that the joint distribution of $a_1, ... a_n$ conditioned on $x_1, ..., x_n$ is given by
\begin{equation}
	p(\bm a|\bm x) = \TrX{}{\left(\bigotimes_{i=1}^nE^i_{a_i| x_i}\right) \Phi^{(t<r)}(\rho^{(0)})} \:,
\end{equation}
where bold font notation is used to indicate vectors of classical data, $\bm x:=(x_1,\ldots,x_n)$ and $\bm a:=(a_1,\ldots,a_n)$, and where
\begin{equation}
	\Phi^{(t<r)} := \Phi^{(t-1)} \circ \ldots \circ \Phi^{(1)} = \left( \Phi^{(t-1)}_{P_1^{(t)}} \otimes ...\otimes \Phi^{(t-1)}_{P_{m_t}^{(t)}}\right) \circ \ldots \circ  \left( \Phi^{(1)}_{P_1^{(1)}} \otimes ...\otimes \Phi^{(1)}_{P_{m_1}^{(1)}}\right) \:.
\end{equation}

Knowing these settings and outcomes, a \emph{behaviour} is the set of conditional distributions
\begin{equation}
	p(\bm a|\bm x)\geq0,\qquad
	\sum_{\bm a\in\prod_i A_i(x_i)}p(\bm a|\bm x)=1 :,
	\label{eq:behavior}
\end{equation}
for all settings and outcomes obtained from a fixed family of settings $\bm X = (X_1,\ldots,X_n)$ and outcomes $\bm A(\bm X) = (A_1(X_1),\ldots,A_n(X_n))$.

For a given behaviour, an assignment of quantum operations to RQT operations is \emph{locality-preserving and context-independent} when the image assigned to a channel or outcome subchannel depends only on that map and its declared input/output blocks, not on the larger protocol or measurement in which it appears, and parallel composites are assigned from their declared components.

The main result of Ref.~\cite{HoffreumonWoods2026} was to prove that a locality-preserving and context-independent assignment exists for every behaviour. There, we proved that the mapping
\begin{subequations}\label{eq:realification_multipartite}
	\begin{align}
		\mathcal{R} : \mathsf{QT} &\longrightarrow \mathsf{RQT} \notag \\ 
		\rho^{(0)} &\longmapsto \tfrac12 \Gamma^{(n)}(\rho^{(0)}) \:, \\
		\forall i:\: E^i_{a_i|x_i} &\longmapsto \Gamma(E^i_{a_i|x_i}) \:,\\
		\forall t, \forall k\in 1, ... m_t\:, \Phi_{P_k}^{(t)} & \longmapsto \Gamma(\Phi_{P_k}^{(t)} )\:,
	\end{align}
\end{subequations}
provides such an assignment because it is based on real-linear maps that preserve the Born rule.

\subsubsection{Locality-preserving real representations of finite protocols}
\label{sec:locality-preserving-protocols}

It remains to adapt the result to superselected theories. We need the following notation:
\begin{equation}
	\mathcal{U}_{A_1A_2... A_n}(g)(\rho) := \left(U_{A_1}(g) \otimes U_{A_2}(g) \otimes ... \otimes U_{A_n}(g)\right) \rho  \left(U_{A_1}(g) \otimes U_{A_2}(g) \otimes ... \otimes U_{A_n}(g)\right)^\dag \:.
\end{equation}

\begin{definition}[Superselected protocols]\label{def:scenario}
	A protocol in $\mathsf{QT}$ is superselected by a group $G$ if, for any behaviour, the initial state and final effects are invariant under the corresponding action of $G$, and every channel used in its description is equivariant under the corresponding action of $G$.
	
	That is, for every $g\in G$,
	\begin{subequations}
		\begin{align}
			&\mathcal{U}_{A_1A_2\ldots A_n}(\rho^{(0)}) = \rho^{(0)}\:,
			\\
			\forall i:\quad &\mathcal{U}_{A_i}(E^i_{a_i|x_i}) = E^i_{a_i|x_i} \:,
			\\
			\forall t, \forall k\in 1, \ldots,m_t:\quad &\mathcal{U}_{P_k^{(t)}} \circ \Phi_{P_k^{(t)}}^{(t)} = \Phi_{P_k^{(t)}}^{(t)} \circ \mathcal{U}_{P_k^{(t-1)}}\:.
		\end{align}
	\end{subequations}
	
\end{definition}

The probabilities, trusted laboratory labels, and asserted partitions are scenario data. Hilbert-space dimensions, concrete matrix realisations of the operators, and the particular tensor decomposition realising those labels are model data. The result below preserves the scenario data as well as the assigned locality of operations; it does not claim that a product or separable quantum initial state is represented by a product or separable real state.

We now apply the locality-preserving realification to finite multipartite protocols in superselected quantum information theory.

\begin{proposition}[Finite protocols]
	\label{thm:adaptive}
	Every finite-dimensional superselected quantum protocol as in
	\cref{def:scenario} has an RQT realization with:
	\begin{enumerate}[label=(\roman*)]
		\item the same joint and conditional outcome probabilities at every round;
		\item the same prescribed locality structure;
		\item a component-wise and context-independent assignment of quantum operations to RQT operations.
	\end{enumerate}
\end{proposition}
\begin{proof}
	This is a direct consequence of \cref{thm:algebra} applied to the result of Ref.~\cite{HoffreumonWoods2026}: since every local operation carries a unitary representation of $G$ (which respects locality when the operation spans multiple subsystems), and since every local operation is mapped to the reals via $\Gamma$ or a map proportional to it, one can locally apply \cref{thm:algebra}, and the result follows.

	The only extra element needed to prove the theorem is a generalisation of \cref{thm:algebra} to the situation in which the multiple-subsystems map $\Gmap^{(n)}$ is used instead of $\Gmap$. The proof is exactly the same once one realises that $\{I^{(n)}_{d_{A_1}d_{A_2}\ldots d_{A_n}},J^{(n)}_{d_{A_1}d_{A_2}\ldots d_{A_n}}\}$ obey the same relations as $\{\id_{2d},J_d\}$, namely, that they provide an orthogonal representation of the group $C_4$.
	
\end{proof}

It follows that not only does every protocol of quantum theory have a real description, but every superselected subtheory of quantum theory also has a real description, even when locality is imposed in the sense of tensor-product factorisation.

\subsection{Example: real fermionic information theory}
\label{sec:real-fermionic-quantum-theory}
To conclude this section, we apply the formalism to fermionic information theory (FIT; a brief review of the elements of the formalism needed for this section is given in \cref{app:complex-fermions}). This example is somewhat peculiar, as fermionic information theory requires substantial modifications compared to more ``usual'' Hilbert-space theories. The main reason for these peculiarities is that the algebra of observables is not taken to be $\cB(\mathcal{H})$, but rather the $\mathbb{Z}_2$-graded algebra generated by the canonical anticommutation relations (CAR).

The $\mathbb{Z}_2$-grading of the algebra leads its associated Hilbert space to decompose into sectors in a particular manner. With such a decomposition, the Hilbert space is usually called a Fock space. In particular, the $\mathbb{Z}_2$ gradation groups the sectors into two classes, called even and odd sectors. The fermionic parity superselection rule is then imposed, forbidding coherent superpositions between these two classes of sectors. Physically, this expresses the fact that fermionic parity is conserved: in particular, fermions can only be created or annihilated in pairs, thereby preventing the conversion of a boson into a fermion.

In this section, we briefly introduce the Fock-space description as a Hilbert-space representation in order to show how the realification map acts on it, and how the parity superselection rule is mapped to a corresponding real superselection rule.

We stress that the word ``real'' in ``real fermionic information theory'' refers only to the RQT representation obtained from the realification construction introduced in the previous sections. In particular, real FIT need not be equivalent to a canonical real representation of the CAR algebra, in the same sense that a real representation of a superselected QT need not, in general, be equivalent to a superselected RQT (although we have shown above that it is). Whether our realification can itself be interpreted as a real representation of the CAR algebra is left for future work.

\subsubsection{Fermionic formalism} 
To make the subsection readable without first consulting the appendix, we briefly reintroduce the concepts needed for the discussion below. Fix $m\in\mathbb N$ to be the number of fermionic modes under consideration, and let $M:=\{1,\ldots,m\}$. The corresponding creation and annihilation operators, denoted $a_j^\dagger$ and $a_j$, satisfy the canonical anticommutation relations (CAR):
\begin{equation}
	\forall i,j\in M:\quad \{a_i,a_j\}=\{a_i^\dagger,a_j^\dagger\}=0 \:, \qquad \{a_i,a_j^\dagger\}=\delta_{ij}\id \:,
\end{equation}
where $\{X,Y\}:=XY+YX$ denotes the anticommutator. The CAR algebra of the full system is the complex algebra generated by these $2m$ operators and the identity:
\begin{equation}
	\mathcal A(M):=\operatorname{alg}_{\C}\{\id,a_j,a_j^\dagger:j\in M\} \:.
\end{equation}

The fermionic Fock space is constructed as the direct sum of antisymmetric tensor powers of the one-particle Hilbert space $\C^m$. Let $\pi\in S_k$ be a permutation of the $k$ tensor factors, and define
\begin{equation}
	\Pi_\pi(v_1\otimes\cdots\otimes v_k) :=v_{\pi^{-1}(1)}\otimes\cdots\otimes v_{\pi^{-1}(k)} 
\end{equation}
as the operator implementing the permutation $\pi$. The antisymmetrizer
\begin{equation}
	A_k:=\frac1{k!}\sum_{\pi\in S_k}\operatorname{sgn}(\pi)\Pi_\pi
\end{equation}
projects $(\C^m)^{\otimes k}$ onto its antisymmetric subspace noted
\begin{equation}
	\operatorname{ran}(A_k)=\wedge^k\C^m \subseteq (\C^m)^{\otimes k} \:.
\end{equation}
For $k\in{1,\ldots,m}$, this is the $k$-particle sector. To complete the picture, the `no-particle' sector, the \textit{vacuum} sector must be added as the one-dimensional space $\wedge^0\C^m\cong\C$. It is spanned by a vector $|\mathrm{vac}\rangle$ satisfying
\begin{equation}
	\forall j\in M:\quad a_j|\mathrm{vac}\rangle=0 \:.
\end{equation}
The Fock space is therefore
\begin{equation}
	\mathcal F_m=\bigoplus_{k=0}^m\wedge^k\C^m\:.
\end{equation}
From the construction, we see that the total dimension of the Fock space, noted $d_m$ must be equal to $2^m$. The subspace $\wedge^k\C^m$ is the eigenspace of the number operator
\begin{equation}
	N:=\sum_{i=1}^m a_i^\dagger a_i
\end{equation}
with eigenvalue $k$. Thus, $m$ counts the total number of fermionic modes, whereas $k$ labels the particle-number sector. We stress this distinction since the fermionic modes themselves do not correspond to tensor-factor subsystems in the locality structure of the CAR algebra.

This distinction is important to define the unrestricted, single-partite fermionic information theory, which is simply single-system QT over Fock spaces. Instead of the systems $A, B, \ldots$, we will use the letter corresponding to the number of modes as a label for various fermionic systems. This gives:
\begin{subequations}\label{eq:FIT}
	\begin{align}
		\operatorname{St}_{\mathsf{FIT}}(m)
		&:=
		\{\rho\in\cB(\mathcal{F}_m):
		\rho\geq0,\ \TrX{}{\rho}=1\}\:,
		\label{eq:FIT-state-set}\\
		\operatorname{Eff}_{\mathsf{FIT}}(m)
		&:=
		\{E\in\cB(\mathcal{F}_m):
		0\leq E\leq\id_{d_m}\}\:,
		\label{eq:FIT-effect-set}\\
		\operatorname{Chan}_{\mathsf{FIT}}(m\!\to\!\ell )
		&:=
		\left\{\Phi:\cB(\mathcal{F}_m)\to\cB(\mathcal{F}_{\ell} ):
		\Phi\ \text{CPTP}\right\}\:.
		\label{eq:FIT-channel-set}
	\end{align}    
\end{subequations}

\paragraph{Parity superselection rule.}
Finally, let $P:=(-1)^N$ be the fermionic parity operator, which has eigenvalue $+1$ on even-particle-number sectors and $-1$ on odd-particle-number sectors. The parity superselection rule forbids coherent superpositions between these two sectors. Equivalently, physical states and effects must be invariant under the parity action:
\begin{equation}
	P\rho P=\rho,\qquad PEP=E,
\end{equation}
or, equivalently, $[\rho,P]=[E,P]=0$. In addition, physical channels are required to be equivariant under this adjoint action. Be aware that the adjoint action of a parity operator is traditionally indicated by $\Theta(X) := PXP$; this is at odds with our notation, which would have marked it as $\mathcal{P}$.

Altogether, this means that the invariant FIT is given by:
\begin{subequations}\label{eq:FIT-inv}
	\begin{align}
		\operatorname{St}_{\mathsf{FIT}^{\mathrm{inv}}_{\mathbb{Z}_2}}(m)
		&:=
		\{\rho\in\cB(\mathcal{F}_m):
		\rho\geq0,\ \TrX{}{\rho}=1,\ [P_m,\rho] = 0\}\:,
		\label{eq:FIT-inv-state-set}\\
		\operatorname{Eff}_{\mathsf{FIT}^{\mathrm{inv}}_{\mathbb{Z}_2}}(m) 
		&:=
		\{E\in\cB(\mathcal{F}_m):
		0\leq E\leq\id_{d_m},\ [P_m,E] = 0\}\:,
		\label{eq:FIT-inv-effect-set}\\
		\operatorname{Chan}_{\mathsf{FIT}^{\mathrm{inv}}_{\mathbb{Z}_2}}(m\!\to\!\ell)
		&:=
		\left\{\Phi:\cB(\mathcal{F}_m)\to\cB(\mathcal{F}_{\ell} ):
		\Phi\ \text{CPTP}, \ \Theta_n \circ \Phi = \Phi \circ \Theta_m \right\}\:.
		\label{eq:FIT-inv-channel-set}
	\end{align}    
\end{subequations}
(In the above, we added an `$m$' subscript to the parity operator $P$ such that the space it is acting on can be distinguished more easily.)   

This is closely analogous to the superselection construction considered for RQT. The parity operator satisfies $P^2=\id$ and therefore provides a unitary representation of the generator of the discrete group $\mathbb Z_2$. In the same way that $J_d^4=\id_{2d}$ provides an orthogonal representation of the generator of $C_4$ in the realification construction, the parity operator implements the $\mathbb Z_2$ symmetry underlying the fermionic superselection rule. Thus, ignoring all the gradation subtleties that induce a Fock-space structure in the Hilbert space and focusing solely on the superselection constraint, fermionic parity can be viewed as restricting the underlying quantum theory to the subtheory invariant under the action of $G=\mathbb Z_2$ under the unitary representation:
\begin{equation}
	\begin{aligned}
		&U: \mathbb{Z}_2 \to U(d_m) \:,\\
		&U(0) = \id_{d_m}\:, \quad U(1) = P_m \:.
	\end{aligned}
\end{equation}

\paragraph{Multipartite fermionic information theory}
To discuss subsystems in fermionic information theory (FIT), one groups the fermionic modes into disjoint \emph{regions}, i.e., disjoint subsets of modes. For $I\subseteq M$, define $\mathcal A(I):=\operatorname{alg}_{\C}\{\id,a_j,a_j^\dagger:j\in I\}$ as the CAR subalgebra associated with region $I$, together with its associated number and parity operators $N_I:=\sum_{j\in I}a_j^\dagger a_j$ and $P_I:=(-1)^{N_I}$.

For two disjoint regions $I,K\subseteq M$, the corresponding CAR subalgebras $\mathcal A(I)$ and $\mathcal A(K)$ satisfy graded commutation relations rather than ordinary commutation relations. In particular, odd elements associated with distinct regions anticommute, while even elements commute. After imposing the parity superselection rule, the physically relevant local observable algebras are the even subalgebras
\begin{equation}
	\mathcal A_{\mathrm{even}}(I):={A\in\mathcal A(I):[A,P_I]=0}\:,
\end{equation}
and similarly for $K$, and these even algebras commute for disjoint regions:
\begin{equation}
	[\mathcal A_{\mathrm{even}}(I),\mathcal A_{\mathrm{even}}(K)]=0\:.
\end{equation}

Thus, parity superselection appears to remove the obstruction to ordinary commutation at the level of physical observables. However, it does not turn the underlying fermionic subsystem structure into an ordinary tensor-product structure. 
In other words, although the superselected algebras of disjoint regions commute, the resulting tensor-product structure fails to properly represent them. This is due to the tensor product not being the natural locality structure of the full fermionic theory, as we detail below.

This distinction is important for the present work. Our realification construction was specifically designed to preserve locality in the strong sense of tensor-product factorisation. Applying it directly to FIT would therefore require imposing a tensor-product notion of locality that is not native to the fermionic theory, one that is already nonlocal at the complex level. We consequently do not pursue a multipartite application of our construction for FIT. The obstruction is discussed at a high level in the remark below; see \cref{subsec:graded-locality-fermionic-quantum-theory} for a more technical account of graded locality and its consequences.

\begin{remark}[Issue with a tensor-product notion of locality in FIT.]
	As we have seen, the graded algebra underlying FIT can ultimately be represented as a direct sum of algebras. This direct sum decomposition is isomorphic to a direct sum of complex-matrix algebras through the Jordan--Wigner transform. Therefore, it can be represented on a tensor-product Hilbert space via a decomposition like that presented in \cref{sec:Tsirelson}. With such a decomposition, one might expect the ``locally acting'' CAR subalgebras to be represented by operators localised on the corresponding tensor factors.
	
	However, the $\mathbb{Z}_2$-grading obstructs such a straightforward notion of tensor-product locality. In particular, odd operators associated with different fermionic subsystems anticommute, whereas operators acting on distinct tensor factors in an ordinary tensor-product Hilbert space commute. Thus, the CAR algebra cannot be faithfully represented with the usual notion of locality in which operators $A,B$ associated with different subsystems take the form $\Tilde{A}\otimes\id$ and $\id\otimes \Tilde{B}$.
	
	The Jordan--Wigner representation resolves this mismatch by introducing a ``Jordan--Wigner string'' into the representation of fermionic operators: a sequence of $\sigma_z$ matrices acting on the preceding tensor factors, together with a local operator acting on the target factor. This string ensures that the resulting operators reproduce the required fermionic anticommutation relations. Consequently, locality in the original graded algebra is not represented as ordinary tensor-product locality in the Jordan--Wigner representation.
	
	For this reason, we will not attempt to apply our multipartite formalism directly to (the Jordan--Wigner representation of) FIT, since our construction was specifically designed to preserve locality in the sense of tensor-product factorisation. Nevertheless, we expect the realification to be able to preserve the structure of the Jordan--Wigner representation, should it be applied to it. Indeed, the Jordan--Wigner strings themselves consist of products of real matrices, and should therefore be preserved under a map of the form $\Gmap\otimes\Gmap\otimes\cdots\otimes\Gmap$.
	
	More generally, the difficulty of using the ordinary tensor product of Hilbert spaces as a notion of parallel composition for a theory whose algebraic composition differs (i.e., the relation $A\otimes B \cong B \otimes A$ is not satisfied by graded operators) is closely related to the issue that motivated our realification map. The parallel composition of QT, namely the \textit{complex} Kronecker product, cannot be directly mapped to the parallel composition of RQT, namely the \textit{real} Kronecker product (i.e., $(cA) \otimes_\C B \cong A \otimes_\C (cB)$ is impossible to represent using $\otimes_\R$ if $c$ is non-real). This motivated Barrios et al.~\cite{Barrios2025} and us~\cite{Hoffreumon2025} to propose the modified $R$-tensor product $\otimes_r$ (or ``flag'' $\otimes_F$ in the nomenclature of Ref.~\cite{Barrios2025}) as an alternative composition rule. Similarly, the graded composition of fermionic systems motivates replacing the ordinary complex tensor product by an appropriate graded tensor product $\widehat{\otimes}$; see, e.g., Ref.~\cite{Vidal_2021}.
\end{remark}

\subsubsection{Realification of unrestricted fermionic information theory}

The realification of fermionic information theory (FIT) follows directly from \cref{thm:algebra}. What remains is to explicitly construct the real fermionic information theory obtained as the image of FIT under the mappings prescribed by the theorem. One reason we chose this example is that the underlying Hilbert space of FIT has enough additional structure that applying the theorem is more than a one-line computation: before applying the theorem, we must identify the underlying Hilbert space and symmetry group, and hence construct the real counterparts of the Fock space and CAR algebra.

Thus, the main task in this section is to construct these realified objects using the map $\Gmap$. The rest is a direct application of the definitions.

We first need to discuss vectors. This requires the vector-wise version of the standard complex-to-real mapping:
\begin{equation}\label{eq:real-fermion-vector-realification}
	\begin{aligned}
		\gamma_d : \C^d &\longrightarrow \R^{2d} \:, \\
		x &\longmapsto \gamma_d(x) = \ket{0}\otimes\ket{\mathrm{Re}(x)} +\ket{1}\otimes\ket{\mathrm{Im}(x)} = 
		\begin{pmatrix}
			\mathrm{Re}(x)\\
			\mathrm{Im}(x)
		\end{pmatrix}\:.
	\end{aligned}
\end{equation}
Once again, the dimension subscript will be omitted when it is clear from context. This map satisfies the following properties for every $\psi,\phi\in\C^d$ and $X\in M_d(\C)$:
\begin{equation}\label{eq:real-fermion-vector-intertwining}
	\begin{aligned}
		&\Gmap(X)\gamma(\psi)=\gamma(X\psi)\:,\\
		&\gamma(\phi)^T\gamma(\psi)=\mathrm{Re}(\phi^\dagger\psi)\:.
	\end{aligned}
\end{equation}
The first identity follows by direct block multiplication using \cref{eq:Gamma}, while the second shows that $\gamma$ is an isometry when the complex Hilbert space is regarded as a real inner-product space.

Since $d_m:=\dim_{\C}(\mathcal F_m)=2^m$ and $\mathcal F_m\cong\C^{d_m}$, its realified image is
\begin{equation}\label{eq:real-fermion-Fock-space}
	\widetilde{\mathcal F}_m:=\gamma(\mathcal F_m) \cong\R^2\otimes\R^{d_m}:.
\end{equation}
This space has real dimension $2d_m=2^{m+1}$.

We define the real annihilation and creation matrices by
\begin{equation}
	\widetilde a_j:=\Gmap(a_j):,\qquad
	\widetilde{a_j^\dagger}:=\Gmap(a_j^\dagger):.
\end{equation}
Because $\Gmap$ is a star homomorphism, cf.~\cref{eq:starhom}, the real creation operator is indeed the transpose of the real annihilation operator:
\begin{equation}
	\widetilde{a_j^\dagger}=\widetilde a_j^T:.
\end{equation}
Moreover, the CAR are faithfully reproduced over the reals:
\begin{equation}\label{eq:real-fermion-CAR}
	\{\widetilde a_i,\widetilde a_j\}=0:,\qquad
	\{\widetilde a_i^T,\widetilde a_j^T\}=0:,\qquad
	\{\widetilde a_i,\widetilde a_j^T\}
	=\delta_{ij}\id_{2d_m}:.
\end{equation}
In the occupation-number basis, the matrices $a_j$ are real (see \cref{eq:fermion-JW-formulas}), and therefore
\begin{equation}\label{eq:real-fermion-generators-block-form}
	\widetilde a_j=I\otimes a_j,\qquad
	\widetilde{a}_j^T=I\otimes a_j^\dagger,\qquad
	[\widetilde a_j,J_{d_m}]
	=[\widetilde{a}_j^T,J_{d_m}]=0:.
\end{equation}
To construct the realified Fock space from its occupation-number eigenstates, we need to mirror the exterior-power construction. We therefore realify the $k$-fold complex tensor space before applying the antisymmetrizer, rather than forming an exterior power directly over $\R$. More precisely, on the space $\gamma\left((\C^m)^{\otimes k}\right)\cong\R^2\otimes\R^{m^k}$ and for $\pi\in S_k$, define
\begin{equation}\label{eq:real-fermion-antisymmetrizer}
	\widetilde{\Pi}_{\pi}:=\Gmap(\Pi_\pi)\:,\qquad
	\widetilde A_k:=\Gmap(A_k)
	=\frac1{k!}\sum_{\pi\in S_k}
	\operatorname{sgn}(\pi)\widetilde{\Pi}_{\pi}:.
\end{equation}
The matrix $\Pi_\pi$ is an element of the permutation group, so it must be real in a given basis and thus $A_k$ is also real in the same basis (see e.g., the tensor basis used in \cref{eq:fermion-antisymmetrizer-matrix}). Consequently,
\begin{equation}\label{eq:real-fermion-antisymmetrizer-properties}
	\widetilde{\Pi}_{\pi}=I\otimes\Pi_\pi,\qquad
	\widetilde A_k=I\otimes A_k\:.
\end{equation}
Moreover, the realified $A_k$ is still an orthogonal projection, $\widetilde A_k^2=\widetilde A_k=\widetilde A_k^T$.

The realified $k$-particle sector must then be built from 
\begin{equation}\label{eq:real-fermion-exterior-sector}
	\widetilde{\wedge^k\C^m} :=\operatorname{ran}(\widetilde A_k) =\gamma(\wedge^k\C^m):.
\end{equation}

We intentionally retain $\C^m$ in the notation $\widetilde{\wedge^k\C^m}$. This is because the realified exterior powers of $\C^m$ are not isomorphic to the exterior powers $\wedge^k\R^{2m}$ obtained by antisymmetrising $\R^{2m}$. Indeed, $\dim_{\C}(\wedge^k\C^m)=\binom{m}{k}$, thus $\dim_{\R}\widetilde{\wedge^k\C^m} =2\binom{m}{k}$. Contrastingly, $\dim_{\R}(\wedge^k\R^{2m})=\binom{2m}{k}$, which implies that the dimensions can only match for $k=1$. Therefore,
\begin{equation}\label{eq:real-fermion-not-real-exterior-power}
	\widetilde{\wedge^k\C^m} \not\cong\wedge^k\R^{2m} \qquad\text{for }k\neq1,
\end{equation}
with the trivial exception $k=1$ understood separately.

This dimensional mismatch can be understood more explicitly by examining the vectors spanning the $k$-particle sectors. Ultimately, it stems from the fact that the image of the complex tensor product $\otimes_\C$ under $\gamma$ is not the ordinary real tensor product $\otimes_\R$, but the flag/$r$-tensor product $\otimes_r$ \cite{Barrios2025,Hoffreumon2025}. Since the exterior product $\wedge$ is constructed from the tensor product $\otimes$, it is therefore unsurprising that the image of $\wedge_\C$ under $\gamma$ does not coincide with $\wedge_\R$.

Let us analyse this difference using the basis vectors of the $k$-particle sectors. Starting with the vacuum sector, $k=0$, we have 
\[
\dim_{\R}(\widetilde{\wedge^0\C^m}) =2\binom{m}{0}=2\:, \qquad
\dim_{\R}(\wedge^0\R^{2m})=\binom{2m}{0}=1\:.
\]
The complex vacuum is therefore realified into a two-dimensional real plane. For $\ket{\mathrm{vac}} \in \mathcal{F}_m \cong \C^{d_m}$, define
\begin{equation}\label{eq:real-fermion-vacuum-plane}
	\begin{gathered}
		|\widetilde{\mathrm{vac}};0\rangle :=\gamma(|\mathrm{vac}\rangle)\:,\qquad
		|\widetilde{\mathrm{vac}};1\rangle :=J_{d_m}|\widetilde{\mathrm{vac}};0\rangle\:, \\
		\widetilde{\mathcal F}_{m,0} :=\operatorname{span}_{\R}
		{|\widetilde{\mathrm{vac}};0\rangle,
			|\widetilde{\mathrm{vac}};1\rangle}
		\cong\R^2:.
	\end{gathered}
\end{equation}
Where the realified $k=0$ sector $\widetilde{\mathcal F}_{m,0}$ is isomorphic to $\widetilde{\wedge^0\C^m}\cong$. Both vacuum vectors are annihilated by every $\widetilde a_j$. This follows directly from the first identity in \eqref{eq:real-fermion-vector-intertwining}, together with \eqref{eq:real-fermion-generators-block-form}.

This is no surprise; recall that a pure state is a complex ray, meaning that it is quotiented by the equivalence relation $\ket{\mathrm{vac}} \sim e^{i\theta} \ket{\mathrm{vac}}$. Over the reals, one needs two dimensions to represent this equivalence relation as the $J_{d_m}$-invariant real plane. This is once again due to the fact that we are working with a real representation of the complex formalism rather than its real counterpart, because our end goal is to prove that any superselected quantum theory can be embedded within a superselected real quantum theory in a way that preserves locality in order to show that any distribution predicted by the former can be reproduced with the latter. At the level of the formalism, superselected quantum theories are obviously inequivalent to superselected real quantum theories.

Hence, these two real vectors represent the two quadratures of a single complex ray, not two distinct physical vacuum states for the embedded theory. Since the density matrix representation is the standard tool to hide away this equivalence relation, we can use it over the reals to hide this cumbersome doubling of the vacuum: the RQT density matrix is no longer a rank-one projector but the rank-2 projector
\begin{equation}\label{eq:real-fermion-vacuum-state}
	\widetilde\rho_{\mathrm{vac}} :=\frac12\Gmap(\ketbra{\mathrm{vac}}{\mathrm{vac}} = =\ket{\gamma(\mathrm{vac})}\otimes_r \bra{\gamma(\mathrm{vac})} 
	=\frac12\sum_{\alpha=0}^1 |\widetilde{\mathrm{vac}};\alpha\rangle \langle\widetilde{\mathrm{vac}};\alpha|.
\end{equation}
The representation condition $[J_{d_m},\widetilde\rho_{\mathrm{vac}}]=0$ is precisely what ensures that $\widetilde\rho_{\mathrm{vac}}$ retains the complex phase structure. By contrast, the naive rank-one projector $ |\gamma(\mathrm{vac})\rangle\langle\gamma(\mathrm{vac})|$ does not commute with $J_{d_m}$ and therefore does not belong to the realified image of the complex state space.

The same doubling occurs in every $k$-particle sector. Let $K\subseteq M$ be a $k$-element subset of modes, say $K=\{j_1,\ldots,j_k\},$ for $1\leq j_1<\cdots<j_k\leq m$, and let $\binom{M}{k}$ denote the set of all $k$-element subsets of $M$. Then
\begin{equation}\label{eq:real-fermion-multi-particle-space}
	\begin{gathered}
		|\widetilde{\bm K};0\rangle :=\widetilde a_{j_1}^\dagger\cdots \widetilde a_{j_k}^\dagger|\widetilde{\mathrm{vac}};0\rangle\:,
		\qquad
		|\widetilde{\bm K};1\rangle :=J_{\binom{m}{k}}|\widetilde{\bm K};0\rangle\:,\\
		\widetilde{\wedge^k\C^m} \cong\widetilde{\mathcal F}_{m,k} :=\operatorname{span}_{\R}\left\{|\widetilde{\bm K};0\rangle, |\widetilde{\bm K};1\rangle: \ K\in\textstyle\binom{M}{k}\right\} \subseteq \R^{2d_m}\:.
	\end{gathered}
\end{equation}
Thus, the mode label $j$ has exactly the same meaning as in the original quantum theory. The additional binary label $\alpha\in{0,1}$ records the two real quadratures associated with complex amplitudes; it does not represent an additional fermionic mode.

It follows by linearity of the realification map that the entire Fock space is
\begin{equation}
	\widetilde{\mathcal F}_m =\bigoplus_{k=0}^m\widetilde{\mathcal F}_{m,k} \cong\bigoplus_{k=0}^m \widetilde{\wedge^k\C^m}\:.
\end{equation}
The different sectors are the eigenspaces of the realified number operator
\begin{equation}
	\widetilde{N} :=\sum_{j\in M}\widetilde a_j^T\widetilde a_j:.
\end{equation}

The realified operator space is subsequently given by
\begin{equation}
	\widetilde{\cB(\mathcal F_m)} :=\Gmap(\cB(\mathcal F_m)) = \{X\in\cB(\widetilde{\mathcal F}_m): [J_{\widetilde{\mathcal F}_m},X]=0\} \cong \{X\in M_{2d_m}(\R): [J_{2d_m},X]=0\} \:.
\end{equation}
For a channel $\Phi : \cB(\mathcal{F}_m) \to \cB(\mathcal{F}_{\ell})$, define its realified counterpart by
\begin{equation}\label{eq:real-fermion-embedded-channel}
	\widetilde{\Phi} : \widetilde{\cB(\mathcal{F}_m)} \to \widetilde{\cB(F_{\ell})} \:,
	\qquad\widetilde\Phi(\Gmap(X)) :=\Gmap(\Phi(X))\:.
\end{equation}
It obeys $\widetilde{\Phi} \circ \mathcal{J}_m = \mathcal{J}_{\ell} \circ \widetilde{\Phi}$ because left- and right-multiplication by $J$ in the real FIT corresponds to multiplication by $i$ before realification.

We can therefore characterise the realified unrestricted fermionic information theory as
\begin{subequations}\label{eq:RFIT}
	\begin{align}
		\operatorname{S}_{\mathsf{FIT},\R}^J(m)
		&:=
		\{\widetilde{\rho}\in\cB(\widetilde{\mathcal F}_m):
		\widetilde{\rho}\geq0,\
		\TrX{}{\widetilde{\rho}}=1,\
		[J_{d_m},\widetilde{\rho}]=0\}\:,
		\label{eq:RFIT-state-set}\\
		\operatorname{E}_{\mathsf{FIT},\R}^J(m)
		&:=
		\{\widetilde E\in\cB(\widetilde{\mathcal F}_m):
		0\leq\widetilde E\leq\id_{d_m},
		[J_{d_m},\widetilde E]=0\}\:,
		\label{eq:RFIT-effect-set}\\
		\operatorname{C}_{\mathsf{FIT},\R}^J(m\!\to\!\ell)
		&:=
		\left\{
		\Phi:\cB(\widetilde{\mathcal F}_m)\to
		\cB(\widetilde{\mathcal F}_{\ell}):
		\Phi\ \text{CPTP},
		\mathcal J_{d_{\ell}}\circ\Phi
		=\Phi\circ\mathcal J_{d_m}
		\right\}\:.
		\label{eq:RFIT-channel-set}
	\end{align}
\end{subequations}

Notice that we have not attempted to show that the realified FIT can itself be understood as a $C_4$-superselected real FIT. Although we know that the corresponding construction holds in general, our reluctance to pursue it here has two main reasons. First, defining what real FIT is requires making choices about what should constitute the real CAR algebras. Second, the dimensional mismatch between $\widetilde{\wedge^k\C^m}$ and $\wedge^k\R^{2m}$ would require reformulating the construction using a factor-wise realification such as $\Gmap^{(k)}$. Addressing these issues in detail would unnecessarily lengthen the article with formally heavier applications of the general machinery to this particular example.

\subsubsection{Realification of the parity superselection rule}
It remains to implement the realified parity superselection rule. Define the parity operator by
\begin{equation}
	\widetilde{P}_m := \Gmap(P_m) \:.
\end{equation}
The induced orthogonal representation of $\mathbb Z_2$ is
\begin{equation}\label{eq:real-fermion-Z2-representation}
	\begin{gathered}
		O_m:\mathbb{Z}_2 \to \mathrm{SO}(2d_m)\:,\\
		O_m(0):=\id_{2d_m},\qquad O_m(1):=\widetilde{P}_m \:.
	\end{gathered}
\end{equation}
Its conjugation-channel representation is
\begin{equation}\label{eq:real-fermion-Z2-channel-representation}
	\mathcal O_m(\epsilon)(X) =O_m(\epsilon)XO_m(\epsilon)^T \qquad(\epsilon\in\{0,1\})\:.
\end{equation}
Since $\mathbb{Z}_2$ is a discrete group with a single generator, invariance under this group,
\begin{equation}
	O_m(\epsilon)(X) = X \quad \forall \epsilon \in \mathbb{Z}_2\:,
\end{equation}
can be simplified into commutation with the generator,
\begin{equation}
	O_m(\epsilon)(X) = X \iff [P_m,X] = 0 \:.
\end{equation}
Also, recall that the conjugation by $P_m$ was denoted $\Theta_m$, accordingly we define
\begin{equation}
	\widetilde{\Theta}_m(X) := \widetilde{P}_m X \widetilde{P}_m \:. 
\end{equation}
Using this later simplification, we can therefore define a characterisation of the parity-superselected real fermionic information theory as
\begin{subequations}\label{eq:RFIT-physcial}
	\begin{align}
		\operatorname{S}_{\mathsf{FIT},\R}^{\mathbb{Z}_2, J}(m)
		&:=
		\{\widetilde{\rho}\in\cB(\widetilde{\mathcal F}_m):
		\widetilde{\rho}\geq0,\
		\TrX{}{\widetilde{\rho}}=1,\
		[J_{d_m},\widetilde{\rho}]=0,\ [\widetilde{P}_m,\widetilde{\rho}]=0\}\:,
		\label{eq:RFIT-physical-state-set}\\
		\operatorname{E}_{\mathsf{FIT},\R}^{\mathbb{Z}_2, J}(m)
		&:=
		\{\widetilde E\in\cB(\widetilde{\mathcal F}_m):
		0\leq\widetilde E\leq\id_{d_m},
		[J_{d_m},\widetilde E]=0, \ [\widetilde{P}_m,\widetilde{E}]=0\}\:,
		\label{eq:RFIT-physical-effect-set}\\
		\operatorname{C}_{\mathsf{FIT},\R}^{\mathbb{Z}_2, J}(m\!\to\!\ell)
		&:=
		\left\{
		\Phi:\cB(\widetilde{\mathcal F}_m)\to
		\cB(\widetilde{\mathcal F}_{\ell}):
		\Phi\ \text{CPTP},
		\mathcal J_{d_{\ell}}\circ\Phi
		=\Phi\circ\mathcal J_{d_m}, \ \widetilde{\Theta}_{\ell} \circ \widetilde{\Phi} = \widetilde{\Phi} \circ \widetilde{\Theta}_m
		\right\}.
		\label{eq:RFIT-physical-channel-set}
	\end{align}
\end{subequations}
The $J_{2d_m}$-commutation condition selects precisely the real matrices in the image of the quantum operator algebra; the $\widetilde{P}_m$-commutation condition is the parity-invariance superselection rule. The condition
involving $J_{2d_m}$ is essential. If it is omitted, one obtains all parity-invariant RQT states and effects, a strictly larger theory than the image of fermionic quantum theory.

Finally, this whole construction is a direct application of \cref{thm:algebra}. For ease of reference, we give the action of the theorem as a proposition rather than a corollary:
\begin{proposition}[Operational equivalence of fermionic information theory and its realification.]
	Defined as an operational probabilistic theory, Fermionic Information Theory has a real image that preserves all probabilities and locality assignments. 
	
	That is, the sets \cref{eq:FIT-inv} are isometrically mapped to the sets \cref{eq:RFIT-physcial} via the triple of mappings $(\frac12 \Gamma, \Gamma, \mathrm{Ad}_\Gamma)$.
\end{proposition}
\begin{proof}
	Follows by construction and \cref{thm:algebra}. The preservation of locality, which in this case corresponds to the (graded) commutation of subalgebras, is a direct consequence of using a $*$-algebra homomorphism for the mapping: every complex (anti)commutator yielding a real operator will have a counterpart in the reals, and $0$ is, in particular, a real operator.
\end{proof}

\section{Conclusion}
We have established a finite-dimensional Tsirelson equivalence for RQT: commuting local measurements and ordinary tensor-product local measurements generate the same correlations. These correlation sets also coincide with their QT counterparts. Thus, in this setting, neither the choice of real or complex scalars nor the choice between these two formulations of measurement locality changes the observable statistics.

We have also shown that finite-dimensional superselected QTs admit RQT representations preserving their symmetry restrictions and all probabilities in finite adaptive protocols, together with the declared tensor-product locality of operations. The resulting RQT realizations are therefore experimentally indistinguishable from the original theories within these protocols. Parity-superselected fermionic QT provides a concrete example. These results extend the operational equivalence of QT and RQT to physically relevant symmetry-restricted theories.

\section*{Acknowledgements}
Proof ideas came from the authors, proof formulation was aided by ChatGTP 5.6 Sol. Writing the manuscript was aided by ChatGTP 5.6 Sol. Authors take full responsibility for proof correctness.

\bibliographystyle{plainnat}
\bibliography{finite_dimensional_tsirelson}

\appendix
\section{Background on fermionic quantum theory}
\label{app:complex-fermions}

We review the conventional finite-dimensional formulation of fermionic
quantum theory, from its global Fock-space and parity-superselection
structure to the graded locality of regional CAR algebras, following
\cite{WWW,BravyiKitaevFermions,DMPT,SzalayFermions}.

\subsection{Global formulation of fermionic quantum theory}
\label{subsec:complex-fermionic-example}
Fix $m$ fermionic modes, with vacuum vector $|\mathrm{vac}\rangle$ and
creation operators $a_1^\dagger,\ldots,a_m^\dagger$.  Here
$a_j^\dagger$ creates one fermion in mode $j$, and $a_j$ is its adjoint
annihilation operator.  They obey the \emph{canonical anticommutation
	relations}, abbreviated \emph{CAR}:
\begin{equation}
	\{a_i,a_j\}=0,\qquad
	\{a_i^\dagger,a_j^\dagger\}=0,\qquad
	\{a_i,a_j^\dagger\}=\delta_{ij}\id,
	\qquad
	a_j|\mathrm{vac}\rangle=0,
	\label{eq:fermion-CAR}
\end{equation}
where $\{X,Y\}:=XY+YX$ is the anticommutator.  Thus two creation operators
anticommute, $a_i^\dagger a_j^\dagger=-a_j^\dagger a_i^\dagger$, and in
particular $(a_j^\dagger)^2=0$, so a single mode cannot be occupied twice.
Accordingly, a \emph{CAR algebra} below means an operator algebra generated by
creation and annihilation operators satisfying these canonical
anticommutation relations.
The vacuum sector is one-dimensional:
\begin{equation}
	\wedge^0\C^m:=\C|\mathrm{vac}\rangle\simeq\C.
	\label{eq:fermion-vacuum-sector}
\end{equation}
This is the $k=0$ sector, i.e. the sector with no particles.  The
one-particle Hilbert space is
\begin{equation}
	\wedge^1\C^m=\C^m
	=\operatorname{span}\{e_1,\ldots,e_m\},
	\qquad
	e_j:=a_j^\dagger|\mathrm{vac}\rangle.
	\label{eq:fermion-one-particle-space}
\end{equation}
Thus $\C^m$ itself describes one fermionic particle, or equivalently one
fermionic excitation above the chosen vacuum, distributed among the $m$
modes.  The basis vector $e_j$ means that mode $j$ is occupied and all other
modes are unoccupied.

For $k\ge1$, the $k$-particle sector is the exterior power
$\wedge^k\C^m$.  It is the antisymmetric subspace generated by occupation of
$k$ distinct modes.  More explicitly, begin with the ordinary tensor-product
space $(\C^m)^{\otimes k}$ for $k$ distinguishable tensor slots.  For a
permutation $\pi\in S_k$, let $P_\pi$ permute those slots:
\begin{equation}
	\Pi_\pi(v_1\otimes\cdots\otimes v_k)
	:=v_{\pi(1)}\otimes\cdots\otimes v_{\pi(k)}.
	\label{eq:fermion-slot-permutation}
\end{equation}
Write $\operatorname{sgn}(\pi)$ for the sign of the permutation:
$\operatorname{sgn}(\pi)=+1$ if $\pi$ is even and
$\operatorname{sgn}(\pi)=-1$ if $\pi$ is odd.  Equivalently, if $\pi$ can be
written as $r$ pairwise swaps, then $\operatorname{sgn}(\pi)=(-1)^r$; this
number is independent of which decomposition into swaps is chosen.
The antisymmetrizer is the projection
\begin{equation}
	A_k:=\frac1{k!}\sum_{\pi\in S_k}\operatorname{sgn}(\pi)\Pi_\pi,
	\qquad A_k^2=A_k=A_k^\dagger,
	\label{eq:fermion-antisymmetrizer}
\end{equation}
and
\begin{equation}
	\wedge^k\C^m:=\operatorname{im}A_k\subseteq(\C^m)^{\otimes k}.
	\label{eq:fermion-exterior-as-image}
\end{equation}
Here $\operatorname{im}A_k$ means the image, or range, of the linear map
$A_k:(\C^m)^{\otimes k}\to(\C^m)^{\otimes k}$:
\begin{equation}
	\operatorname{im}A_k
	:=\{A_kx:\ x\in(\C^m)^{\otimes k}\}.
	\label{eq:fermion-image-definition}
\end{equation}
Thus $\wedge^k\C^m$ consists precisely of those tensors obtained by applying
the antisymmetrizing projection $A_k$ to arbitrary $k$-slot tensors.
This gives a concrete matrix representation: in the tensor-product basis
$e_{i_1}\otimes\cdots\otimes e_{i_k}$, each $P_\pi$ is an
$m^k\times m^k$ permutation matrix, and $A_k$ is an $m^k\times m^k$
Hermitian projection matrix.  Equivalently, its matrix entries are
\begin{equation}
	(A_k)_{\bm i,\bm j}
	=\frac1{k!}\sum_{\pi\in S_k}\operatorname{sgn}(\pi)
	\prod_{r=1}^k\delta_{i_r,j_{\pi(r)}}.
	\label{eq:fermion-antisymmetrizer-matrix}
\end{equation}
In the ordered mode basis, $\wedge^k\C^m$ has basis
\begin{equation}
	e_{j_1}\wedge\cdots\wedge e_{j_k},
	\qquad 1\le j_1<\cdots<j_k\le m,
	\label{eq:fermion-wedge-basis}
\end{equation}
with antisymmetry
$e_p\wedge e_q=-e_q\wedge e_p$ and $e_p\wedge e_p=0$.  With the normalized
convention used here,
\begin{equation}
	e_{j_1}\wedge\cdots\wedge e_{j_k}
	:=\frac1{\sqrt{k!}}\sum_{\pi\in S_k}\operatorname{sgn}(\pi)
	e_{j_{\pi(1)}}\otimes\cdots\otimes e_{j_{\pi(k)}}.
	\label{eq:fermion-normalized-wedge}
\end{equation}
This basis vector is the occupation-number state
\begin{equation}
	e_{j_1}\wedge\cdots\wedge e_{j_k}
	\longleftrightarrow
	a_{j_1}^\dagger\cdots a_{j_k}^\dagger|\mathrm{vac}\rangle
	=|n_1,\ldots,n_m\rangle,
	\label{eq:fermion-occupation-state}
\end{equation}
where $n_j=1$ exactly for $j\in\{j_1,\ldots,j_k\}$ and $n_j=0$ otherwise.
Thus, in a spatial site basis, this may be pictured as $k$ fermionic
excitations occupying $m$ possible locations, with at most one excitation in
each mode because $n_j\in\{0,1\}$.  More generally, however, a \emph{mode}
means a chosen one-particle degree of freedom: it may be a location, but it
may also be an orbital, spin-orbital, momentum mode, or other one-particle
basis state.  For example, two spin states at the same spatial site are two
distinct fermionic modes.

The full fermionic Fock space allows a superposition or mixture of different
particle numbers:
\begin{equation}
	\mathcal F_m=\bigoplus_{k=0}^m\wedge^k\C^m
	=\mathcal F_m^{\bar 0}\oplus\mathcal F_m^{\bar 1}.
	\label{eq:fermion-fock-decomposition}
\end{equation}
The upper limit is $m$ because, by $e_j\wedge e_j=0$, no mode can be occupied
twice.  The $k$th summand $\wedge^k\C^m$ has definite particle number $k$;
the parity superselection rule below is an additional rule about which
coherences between different $k$-sectors are physically allowed.
The even- and odd-particle-number subspaces are
explicitly
\begin{equation}
	\begin{aligned}
		\mathcal F_m^{\bar 0}
		&:=\bigoplus_{\substack{0\le k\le m\\ k\ {\rm even}}}\wedge^k\C^m,
		&
		\mathcal F_m^{\bar 1}
		&:=\bigoplus_{\substack{0\le k\le m\\ k\ {\rm odd}}}\wedge^k\C^m.
	\end{aligned}
	\label{eq:fermion-parity-sectors}
\end{equation}
Let
$N=\sum_{j=1}^m a_j^\dagger a_j$ be the number operator and let
\begin{equation}
	P:=(-1)^N
	\label{eq:fermion-parity-operator}
\end{equation}
be total fermion parity.  The parity superselection rule says that physical
states and effects are even, i.e. they commute with $P$; physical channels
are equivariant for parity conjugation and therefore preserve the even operator
subalgebra.  An equivariant channel need not preserve the even and odd Hilbert
subspaces separately (an odd unitary, for example, can exchange them).  This
is exactly the invariant world for the group
$G=\mathbb Z_2=\{0,1\}$ represented by
\begin{equation}
	U(0)=\id_{\mathcal F_m},\qquad U(1)=P.
	\label{eq:fermion-Z2-representation}
\end{equation}
Its unitary conjugation-channel representation is
\begin{equation}
	\mathcal U(\epsilon)(X)
	:=U(\epsilon)XU(\epsilon)^\dagger
	\qquad(\epsilon\in\mathbb Z_2).
	\label{eq:fermion-Z2-channel-representation}
\end{equation}
Indeed, \eqref{eq:invworldstate}--\eqref{eq:invworldchannel} become
\begin{equation}
	\mathcal U(1)(\rho)=\rho,\qquad
	\mathcal U(1)(E)=E,\qquad
	\Phi\circ\mathcal U(1)=\mathcal U(1)\circ\Phi.
	\label{eq:fermion-invariant-conditions}
\end{equation}
Equivalently, no physical state or measurement contains coherence between
$\mathcal F_m^{\bar 0}$ and $\mathcal F_m^{\bar 1}$, and dynamics is
equivariant for the parity action.  This is the standard parity-superselected
formulation of finite fermionic quantum theory: Wick, Wightman, and Wigner
introduced superselection rules as restrictions on possible measurements and
proved such restrictions for spinor fields \cite{WWW}; Bravyi and Kitaev use
local fermionic modes as elementary computational systems
\cite{BravyiKitaevFermions}; and D'Ariano, Manessi, Perinotti, and Tosini
identify the operational fermionic model's nonlocal-tomographic features as
direct consequences of parity superselection \cite{DMPT}.

\subsection{Graded locality in fermionic quantum theory}
\label{subsec:graded-locality-fermionic-quantum-theory}

For many mode regions, the phrase ``fermionic locality'' needs a little care.  We now spell
out the graded CAR locality structure and the relation to qubits; see
\cite[Secs.~2.2--3.2, Eq.~(88), and Sec.~7.1]{SzalayFermions}.

Here \emph{graded} is a technical term.  A $\mathbb Z_2$-grading assigns to
each homogeneous operator one of two grades,
$\bar0$ (even) or $\bar1$ (odd), and the grade of a product is the sum of the
grades modulo two.  In the present fermionic algebra, a monomial containing an
even number of creation and annihilation operators has grade $\bar0$, while
one containing an odd number has grade $\bar1$.  Intuitively, the grading is
bookkeeping for fermionic signs: passing one odd operator through another odd
operator contributes a minus sign, whereas passing an even operator through
anything contributes no sign.  The phrase \emph{graded locality} is used
because the commutation rule for operators in disjoint regions depends on
these grades, rather than all disjoint operators simply commuting.  The
following definitions make this statement precise.

\emph{Regional CAR algebras.}
Let $M:=\{1,\ldots,m\}$ be the set of mode labels and let $I\subseteq M$ be a
region, meaning a chosen subset of modes.  Its regional CAR algebra is the
unital $*$-subalgebra
\begin{equation}
	\mathcal A(I)
	:=\operatorname{Alg}^{*}\!\left(
	\id,\{a_j,a_j^\dagger:j\in I\}\right)
	\subseteq\mathcal B(\mathcal F_m).
	\label{eq:fermion-regional-CAR-algebra}
\end{equation}
Here $\mathcal B(\mathcal F_m)$ is the algebra of all linear operators on
$\mathcal F_m$, and $\operatorname{Alg}^{*}(S)$ means the smallest complex
linear subspace containing the set $S$ that is closed under operator
multiplication and adjoint.  Equivalently, $\mathcal A(I)$ consists of linear
combinations of products, or \emph{monomials}, in the $a_j$ and $a_j^\dagger$
with $j\in I$, together with the identity.  Thus localization refers to the
mode labels appearing in the CAR generators, not initially to a tensor factor
of a qubit Hilbert space.

Put
\begin{equation}
	N_I:=\sum_{j\in I}a_j^\dagger a_j,\qquad
	P_I:=(-1)^{N_I},\qquad
	\Theta_I(X):=P_I X P_I
	\quad (X\in\mathcal A(I)).
	\label{eq:fermion-local-parity-grading}
\end{equation}
$N_I$ counts the occupied modes in $I$, $P_I$ is the parity of that number,
and $\Theta_I$ is the corresponding parity automorphism of
$\mathcal A(I)$.  For $p\in\{0,1\}$, define the homogeneous parity subspace
\begin{equation}
	\mathcal A(I)^{\bar p}
	:=\{X\in\mathcal A(I):\Theta_I(X)=(-1)^pX\}.
	\label{eq:fermion-homogeneous-subspaces}
\end{equation}
An element of $\mathcal A(I)^{\bar0}$ is called \emph{even}, and an element of
$\mathcal A(I)^{\bar1}$ is called \emph{odd}.  A monomial is even or odd
according as it contains an even or odd total number of creation and
annihilation operators.  Every $X\in\mathcal A(I)$ has the unique
decomposition
\begin{equation}
	X=X_{\bar0}+X_{\bar1},\qquad
	X_{\bar p}:=\frac12\bigl(X+(-1)^p\Theta_I(X)\bigr)
	\in\mathcal A(I)^{\bar p}.
	\label{eq:fermion-parity-decomposition}
\end{equation}
Moreover, multiplication respects these two grades:
\begin{equation}
	\mathcal A(I)^{\bar p}\mathcal A(I)^{\bar q}
	\subseteq\mathcal A(I)^{\overline{p+q}},
	\qquad p+q\ \text{computed modulo }2.
	\label{eq:fermion-grading-product-rule}
\end{equation}
Equations \eqref{eq:fermion-parity-decomposition} and
\eqref{eq:fermion-grading-product-rule} are the precise statement that
$\mathcal A(I)$ is a $\mathbb Z_2$-graded algebra.

\emph{Graded locality.}
If $I,J\subseteq M$ are disjoint and
$X\in\mathcal A(I)^{\bar p}$ and
$Y\in\mathcal A(J)^{\bar q}$ are homogeneous, the CAR
\eqref{eq:fermion-CAR} give
\begin{equation}
	XY=(-1)^{pq}YX,\qquad p,q\in\{0,1\}.
	\label{eq:fermion-graded-commutation}
\end{equation}
The factor $(-1)^{pq}$ is called the \emph{Koszul sign} associated with
exchanging homogeneous elements of parity degrees $p$ and $q$.  It equals
$+1$ if at least one element is even and equals $-1$ precisely when both
elements are odd.  Thus interchanging two odd fermionic objects introduces a
minus sign, whereas every other interchange introduces no sign.
Indeed, for monomials, moving every CAR generator in $X$ past every generator
in $Y$ produces one minus sign for each crossing; the total sign is
$(-1)^{(\deg X)(\deg Y)}=(-1)^{pq}$.  Linearity then proves
\eqref{eq:fermion-graded-commutation} for arbitrary homogeneous elements.
Consequently, two odd operators in disjoint regions anticommute, whereas two
disjoint operators commute whenever at least one of them is even.  For
example, $a_i a_j=-a_j a_i$ for $i\ne j$, but the even number operator
$a_i^\dagger a_i$ commutes with both $a_j$ and $a_j^\dagger$.

Equation \eqref{eq:fermion-graded-commutation} is encoded algebraically by the
\emph{graded tensor product}, denoted $\widehat\otimes$.  On homogeneous
elements its multiplication and adjoint are defined by
\begin{equation}
	\begin{aligned}
		(X_{\bar p}\widehat\otimes Y_{\bar q})
		(X'_{\bar p'}\widehat\otimes Y'_{\bar q'})
		&:=(-1)^{qp'}
		(X_{\bar p}X'_{\bar p'})
		\widehat\otimes
		(Y_{\bar q}Y'_{\bar q'}),\\
		(X_{\bar p}\widehat\otimes Y_{\bar q})^\dagger
		&:=(-1)^{pq}
		X_{\bar p}^\dagger\widehat\otimes Y_{\bar q}^\dagger,
	\end{aligned}
	\label{eq:fermion-graded-tensor-product}
\end{equation}
and these rules are extended linearly.  The sign in the first line records
the passage of $Y_{\bar q}$ past $X'_{\bar p'}$.  With these definitions the
map
\begin{equation}
	\mathcal A(I)\widehat\otimes\mathcal A(J)
	\longrightarrow\mathcal A(I\cup J),
	\qquad X\widehat\otimes Y\longmapsto XY,
	\label{eq:fermion-graded-product-map}
\end{equation}
is a $*$-algebra isomorphism for disjoint finite mode sets $I$ and $J$.
This is the precise sense in which disjoint fermionic regions compose by a
graded, rather than an ordinary, tensor product.

\emph{Why the global parity representation factorizes.}
The graded composition law for the full CAR algebra does not prevent the
fermion-parity symmetry itself from factorizing.  Let
\begin{equation}
	M=I_1\mathbin{\dot\cup}\cdots\mathbin{\dot\cup}I_r,
	\qquad
	\mathfrak h_{I_\ell}:=\operatorname{span}\{e_j:j\in I_\ell\},
	\label{eq:fermion-mode-region-partition}
\end{equation}
be a partition of the modes.  Here $\mathbin{\dot\cup}$ denotes
\emph{disjoint union}.  The regional number and
parity operators were
defined in \eqref{eq:fermion-local-parity-grading} as
$N_{I_\ell}=\sum_{j\in I_\ell}a_j^\dagger a_j$ and
$P_{I_\ell}=(-1)^{N_{I_\ell}}$.  Since the regional number operators commute
and sum to the total number operator,
\begin{equation}
	[N_{I_\ell},N_{I_s}]=0,
	\qquad
	N=\sum_{\ell=1}^rN_{I_\ell},
	\label{eq:fermion-number-region-sum}
\end{equation}
one finds
\begin{equation}
	P=(-1)^N=e^{i\pi N}
	=\prod_{\ell=1}^r e^{i\pi N_{I_\ell}}
	=\prod_{\ell=1}^rP_{I_\ell}.
	\label{eq:fermion-parity-region-product}
\end{equation}

Equivalently, the one-particle decomposition
$\C^m=\bigoplus_{\ell=1}^r\mathfrak h_{I_\ell}$ induces the graded Fock-space
factorization
\begin{equation}
	\mathcal F(\C^m)
	\simeq
	\widehat\bigotimes_{\ell=1}^r\mathcal F(\mathfrak h_{I_\ell}),
	\qquad
	P\simeq\widehat\bigotimes_{\ell=1}^rP_{I_\ell}.
	\label{eq:fermion-graded-parity-factorization}
\end{equation}
Every $P_{I_\ell}$ is even: it is a function of the even operator
$N_{I_\ell}$ and satisfies
$\Theta_{I_\ell}(P_{I_\ell})=P_{I_\ell}$ by the definition of
$\Theta_{I_\ell}$ in \eqref{eq:fermion-local-parity-grading}.  Therefore no Koszul minus sign
appears when the factors in \eqref{eq:fermion-graded-parity-factorization}
are exchanged.  Fix the block order $I_1,\ldots,I_r$ and, inside each block,
the order inherited from $1<\cdots<m$.  If
$\mathbf n_{I_\ell}:=(n_j)_{j\in I_\ell}$ in that inherited order, define
the ordered Jordan--Wigner identification by
\begin{equation}
	W_{\mathsf P}:\mathcal F(\C^m)\longrightarrow
	\bigotimes_{\ell=1}^r\mathcal F(\mathfrak h_{I_\ell}),
	\qquad
	W_{\mathsf P}|n_1,\ldots,n_m\rangle
	:=\bigotimes_{\ell=1}^r|\mathbf n_{I_\ell}\rangle.
	\label{eq:fermion-block-JW-identification}
\end{equation}
It is unitary because it maps one orthonormal occupation-number basis
bijectively onto another.  Under this identification the graded Fock product
is represented on an ordinary tensor product, and one obtains the ordinary
matrix (Kronecker-product) factorization
\begin{equation}
	W_{\mathsf P}PW_{\mathsf P}^\dagger
	=\bigotimes_{\ell=1}^rP_{I_\ell}.
	\label{eq:fermion-ordinary-parity-factorization}
\end{equation}
In particular, for one mode in each block, the right-hand side is
$Z^{\otimes m}$, where $Z:=\operatorname{diag}(1,-1)$.

Thus the global $\mathbb Z_2$ representation is an on-site, factorized
representation.  For $\epsilon\in\mathbb Z_2=\{0,1\}$, define
\begin{equation}
	U_{I_\ell}(\epsilon):=P_{I_\ell}^{\epsilon};
	\qquad
	U(\epsilon):=P^\epsilon
	\simeq\widehat\bigotimes_{\ell=1}^rU_{I_\ell}(\epsilon),
	\label{eq:fermion-Z2-factorized-representation}
\end{equation}
and, after the ordered identification, the last graded tensor product is the
ordinary Kronecker product.  The symmetry is called \emph{global} because the
same group element $\epsilon$ acts on every region, not because
$U(\epsilon)$ is an entangling operator.  The same calculation gives the
factorization of number rotations,
\begin{equation}
	e^{-i\theta N}
	=\prod_{\ell=1}^re^{-i\theta N_{I_\ell}}
	\simeq\widehat\bigotimes_{\ell=1}^r e^{-i\theta N_{I_\ell}}.
	\label{eq:fermion-U1-factorized-representation}
\end{equation}

This statement must not be confused with ordinary locality of all fermionic
operators.  Odd regional operators still satisfy
\eqref{eq:fermion-graded-commutation} and acquire Jordan--Wigner strings; it
is the \emph{even parity representation} that factorizes without signs.  The
multipartite theorem \cref{thm:adaptive} applies directly when the chosen
superselection rule requires states and effects to be invariant under parity,
number rotations, or another block-preserving on-site symmetry; no action on
classical measurement labels is introduced.

\emph{What the Jordan--Wigner identification does.}
Choose the ordering $1<\cdots<m$ and define the unitary identification
\begin{equation}
	W:\mathcal F_m\longrightarrow(\C^2)^{\otimes m},
	\qquad
	W|n_1,\ldots,n_m\rangle
	:=|n_1\rangle\otimes\cdots\otimes|n_m\rangle.
	\label{eq:fermion-JW-unitary}
\end{equation}
In the one-qubit basis $\{|0\rangle,|1\rangle\}$, put
\begin{equation}
	Z:=|0\rangle\!\langle0|-|1\rangle\!\langle1|
	=\begin{pmatrix}1&0\\0&-1\end{pmatrix},\quad
	\sigma^-:=|0\rangle\!\langle1|
	=\begin{pmatrix}0&1\\0&0\end{pmatrix},\quad
	\sigma^+:=|1\rangle\!\langle0|
	=\begin{pmatrix}0&0\\1&0\end{pmatrix}.
	\label{eq:fermion-JW-one-qubit-matrices}
\end{equation}
Thus $\sigma^-$ changes an occupied qubit to an unoccupied qubit,
$\sigma^+$ does the reverse, and $Z$ supplies the fermionic sign.  The
Jordan--Wigner formulas are
\begin{equation}
	\begin{aligned}
		W a_j W^\dagger
		&=Z^{\otimes(j-1)}\otimes\sigma^-
		\otimes\id^{\otimes(m-j)},\\
		W a_j^\dagger W^\dagger
		&=Z^{\otimes(j-1)}\otimes\sigma^+
		\otimes\id^{\otimes(m-j)}.
	\end{aligned}
	\label{eq:fermion-JW-formulas}
\end{equation}
The preceding string of $Z$ operators is what makes generators belonging to
different modes anticommute.

For the complete mode set, the CAR representation is irreducible, and hence
\begin{equation}
	W\mathcal A(M)W^\dagger
	=\mathcal B\bigl((\C^2)^{\otimes m}\bigr)
	\simeq M_{2^m}(\C).
	\label{eq:fermion-full-algebra-qubit-isomorphism}
\end{equation}
This is exactly what ``the full fermionic operator algebra is qubit-like''
means: as an abstract finite-dimensional $*$-algebra, it is the same full
matrix algebra as that of $m$ qubits.  It does \emph{not} mean that regional
algebras have the ordinary qubit embedding.  Indeed,
\begin{equation}
	W\mathcal A(I)W^\dagger
	=\operatorname{Alg}^{*}\!\left(
	\id,\{
	Z^{\otimes(j-1)}\otimes\sigma^\pm
	\otimes\id^{\otimes(m-j)}:j\in I\}\right).
	\label{eq:fermion-JW-regional-algebra}
\end{equation}
For example, the odd operator of mode $2$ is represented as
$W(a_2+a_2^\dagger)W^\dagger=Z\otimes
(\sigma^-+\sigma^+)\otimes\id^{\otimes(m-2)}$; it has a $Z$ string on the
earlier mode and is therefore not an operator on the second qubit tensor
factor alone.  Changing the mode ordering moves these strings but does not
turn the regional CAR algebras for all possible mode subsets simultaneously
into the usual qubit-local tensor factors.

\emph{Effect of parity superselection on locality.}
A physical observable localized in $I$ is required to be locally even, hence
belongs to $\mathcal A(I)^{\bar0}$; analogously, a physical local channel
$\Phi_I$ is parity equivariant,
$\Phi_I\circ\Theta_I=\Theta_I\circ\Phi_I$.  It follows immediately from
\eqref{eq:fermion-graded-commutation} that physical observables in disjoint
regions commute.  For a partition $\mathsf P=\{I_1,\ldots,I_r\}$ of $M$, the
algebra generated by locally measurable observables is therefore
\begin{equation}
	\mathcal A_{\mathsf P}^{\mathrm{loc}}
	:=\operatorname{Alg}^{*}\!\left(
	\bigcup_{\ell=1}^r\mathcal A(I_\ell)^{\bar0}\right)
	\simeq
	\bigotimes_{\ell=1}^r\mathcal A(I_\ell)^{\bar0},
	\label{eq:fermion-locally-even-algebra}
\end{equation}
where the tensor product on the right is now the ordinary tensor product of
mutually commuting even algebras.

\emph{Local POVMs explicitly.}
Let $x$ be a classical measurement-setting label for region $I$, and let
$\Omega_x$ be its finite set of outcome labels.  A parity-respecting local
POVM for setting $x$ is a family of effects
$\{E_{a|x}:a\in\Omega_x\}$ satisfying
\begin{equation}
	E_{a|x}\in\mathcal A(I)^{\bar0},
	\qquad E_{a|x}\geq0,
	\qquad \sum_{a\in\Omega_x}E_{a|x}=\id.
	\label{eq:fermion-local-even-POVM}
\end{equation}
The first condition is the parity-superselection requirement.  In terms of
the regional parity operator, it says equivalently that every individual
outcome effect is invariant under the parity automorphism:
\begin{equation}
	\Theta_I(E_{a|x})=P_I E_{a|x}P_I=E_{a|x}
	\quad\Longleftrightarrow\quad
	[E_{a|x},P_I]=0.
	\label{eq:fermion-POVM-parity-invariance}
\end{equation}
Thus the superselection rule is imposed outcome by outcome, rather than only
on the sum of the POVM effects (whose sum is already the identity).

Now let $J$ be disjoint from $I$, and let
$\{F_{b|y}:b\in\Omega_y\}\subset\mathcal A(J)^{\bar0}$ be another local
parity-respecting POVM.  Applying
\eqref{eq:fermion-graded-commutation} with parity degrees $p=q=0$ gives
\begin{equation}
	E_{a|x}F_{b|y}
	=(-1)^{0\cdot0}F_{b|y}E_{a|x}
	=F_{b|y}E_{a|x},
	\qquad [E_{a|x},F_{b|y}]=0.
	\label{eq:fermion-disjoint-even-effects-commute}
\end{equation}
Their joint local measurement is therefore the POVM with effects
\begin{equation}
	G_{a,b|x,y}:=E_{a|x}F_{b|y}.
	\label{eq:fermion-product-POVM-effects}
\end{equation}
Indeed, the two positive factors commute, so functional calculus gives
$G_{a,b|x,y}=E_{a|x}^{1/2}F_{b|y}E_{a|x}^{1/2}\geq0$, and normalization follows
from
\begin{equation}
	\sum_{a\in\Omega_x}\sum_{b\in\Omega_y}G_{a,b|x,y}
	=\left(\sum_aE_{a|x}\right)
	\left(\sum_bF_{b|y}\right)
	=\id.
	\label{eq:fermion-product-POVM-normalization}
\end{equation}
For a global state $\rho$, the corresponding joint outcome statistics are
therefore
\begin{equation}
	p(a,b|x,y)=\TrX{}{\rho E_{a|x}F_{b|y}}.
	\label{eq:fermion-local-POVM-joint-statistics}
\end{equation}

The product in \eqref{eq:fermion-product-POVM-effects} is the intrinsic CAR
description.  Under the ordinary tensor-product isomorphism in
\eqref{eq:fermion-locally-even-algebra}, the multiplication map satisfies
\begin{equation}
	\mu_{I,J}:\mathcal A(I)^{\bar0}\otimes\mathcal A(J)^{\bar0}
	\longrightarrow
	\operatorname{Alg}^{*}\!\bigl(
	\mathcal A(I)^{\bar0}\cup\mathcal A(J)^{\bar0}\bigr),
	\qquad
	\mu_{I,J}(E_{a|x}\otimes F_{b|y})=E_{a|x}F_{b|y}.
	\label{eq:fermion-even-product-isomorphism}
\end{equation}
Here $\otimes$ is the ordinary tensor product and, after bases are chosen,
its matrix realization is the Kronecker product.  This conclusion concerns
the locally even observable algebras only: it does not turn odd CAR operators
into ordinary local tensor-factor operators, because those operators still
obey the graded commutation law and carry Jordan--Wigner strings.

There remains an important distinction between \emph{locally even} and
\emph{globally even}.  The global physical algebra
$\mathcal A(M)^{\bar0}$ also contains products
$XY$ with $X\in\mathcal A(I)^{\bar1}$ and
$Y\in\mathcal A(J)^{\bar1}$: each factor is odd, but their product is even.
Thus, for a partition with at least two nonempty parts, generally
\begin{equation}
	\mathcal A_{\mathsf P}^{\mathrm{loc}}
	\subsetneq\mathcal A(M)^{\bar0}.
	\label{eq:fermion-local-global-even-inclusion}
\end{equation}
Consequently, statistics of products of locally even measurements need not
determine all globally even correlations; this is the origin of the
nonlocal-tomographic behavior discussed in \cite{DMPT,SzalayFermions}.

We can now state the two differences from ordinary qubits precisely:
although \eqref{eq:fermion-full-algebra-qubit-isomorphism} identifies the
full matrix algebras, (i) the regional inclusions
$I\mapsto\mathcal A(I)$ obey graded commutation and acquire
Jordan--Wigner strings rather than the usual qubit tensor-factor inclusion,
and (ii) parity superselection restricts physical states and observables to
the even part and channels to the equivariant part.  The present
$\mathbb Z_2$ example captures (ii).  Forgetting the symmetry restriction
leaves every regional algebra and every grading unchanged, because it changes
no operator.

Ordinary tensor-factor locality should, however, be distinguished from CAR
locality.  For a fermionic application, the regional net
$I\mapsto\mathcal A(I)$ (or, equivalently, the graded tensor product
\eqref{eq:fermion-graded-tensor-product}) is additional trusted locality data.

\end{document}